\documentclass[12pt]{article}

\pdfoutput=1

\usepackage{amsthm,enumerate}
\usepackage{mathtools}
\usepackage[normalem]{ulem}
\usepackage[T1]{fontenc}
\usepackage[utf8]{inputenc}
\usepackage[thinc]{esdiff}
\usepackage{color}
\usepackage{fancyhdr}
\usepackage{calc}
\usepackage{url}
\usepackage{xcolor}
\usepackage{xspace}

\usepackage[top=80pt,bottom=85pt,left=85pt,right=85pt]{geometry}
\usepackage{amsmath,amssymb,graphicx,float,color,tikz,cite,savesym,wasysym,subfigure}
\usepackage{caption}
\usepackage[debug,pageanchor=false]{hyperref}
\definecolor{link}{rgb}{.8,.15,.1}
\hypersetup{colorlinks=true,linkcolor=link,citecolor=link,urlcolor=link,linktocpage}
   
\newcommand{\R}{\mathbb{R}}

\newcommand{\ii}{\mathrm{i}}
\newcommand{\dd}{\mathrm{d}}
\newcommand{\ee}{\mathrm{e}}

\newcommand{\N}{\mathbb{N}}

\newcommand{\BV}{\text{\rm BV}}

\newcommand{\supp}{\text{\rm supp}}

\newcommand{\Lip}{\mathrm{Lip}}

\newcommand{\RCD}{\mathsf{RCD}}
\newcommand{\CD}{\mathsf{CD}}

\DeclareMathOperator*{\sgn}{sgn}

\newcommand{\mm}{\mathfrak m}
\newcommand{\di}{\mathsf{d}}

\newcommand{\sfd}{\mathsf d}

\newcommand{\Ch}{\mathsf{Ch}}
\newcommand{\Per}{\mathsf{Per}}

\theoremstyle{plain}
\newtheorem{lemma}{Lemma}[section]
\newtheorem{theorem}[lemma]{Theorem}

\newtheorem*{theorem*}{Theorem}
\newtheorem*{maintheorem*}{Main Theorem}

\theoremstyle{definition}

\newtheorem{definition}[lemma]{Definition}
\newtheorem*{definition*}{Definition}
\newtheorem*{remark*}{Remark}
\newtheorem{remark}[lemma]{Remark}

\makeatletter
\@addtoreset{equation}{section}
\makeatother

\begin{document}

	\begin{titlepage}

	\begin{center}

	\vskip .5in 
	\noindent

	{\Large \bf{Cheeger bounds on spin-two fields}}

	\bigskip\medskip
	 G. Bruno De Luca,$^1$  Nicol\`o De Ponti,$^2$
	Andrea Mondino,$^3$	Alessandro Tomasiello$^{4}$\\

	\bigskip\medskip
	{\small 
$^1$ Stanford Institute for Theoretical Physics, Stanford University,\\
382 Via Pueblo Mall, Stanford, CA 94305, United States
\\	
	\vspace{.3cm}
	$^2$ 
	International School for Advanced Studies (SISSA),\\
	Via Bonomea 265, 34136 Trieste, Italy
\\	
	\vspace{.3cm}
	$^3$ Mathematical Institute, University of Oxford, Andrew-Wiles Building,\\ Woodstock Road, Oxford, OX2 6GG, UK
\\
	\vspace{.3cm}
$^4$ Dipartimento di Matematica, Universit\`a di Milano--Bicocca, \\ Via Cozzi 55, 20126 Milano, Italy \\ and \\ INFN, sezione di Milano--Bicocca
		}

   \vskip .5cm 
	{\small \tt gbdeluca@stanford.edu, ndeponti@sissa.it,\\ andrea.mondino@maths.ox.ac.uk, alessandro.tomasiello@unimib.it}
	\vskip .9cm 
	     	{\bf Abstract }
	\vskip .1in
	\end{center}

	\noindent

We consider gravity compactifications whose internal space consists of small bridges connecting larger manifolds, possibly noncompact. We prove that, under rather general assumptions, this leads to a massive spin-two field with very small mass. The argument involves a recently-noticed relation to Bakry--\'Emery geometry, a version of the so-called Cheeger constant, and the theory of synthetic Ricci lower bounds. The latter technique allows generalizations to non-smooth spaces such as those with D-brane singularities. For AdS$_d$ vacua with a bridge admitting an AdS$_{d+1}$ interpretation, the holographic dual is a CFT$_d$ with two CFT$_{d-1}$ boundaries. The ratio of their degrees of freedom gives the graviton mass, generalizing results obtained by Bachas and Lavdas for $d=4$. We also prove new bounds on the higher eigenvalues. 
These are in agreement with the spin-two swampland conjecture in the regime where the background is scale-separated; in the opposite regime we provide examples where they are in naive tension with it.

	\noindent

	\vfill
	\eject

	\end{titlepage}
    
\tableofcontents

\section{Introduction} 
\label{sec:intro}

In a gravity compactification, the Kaluza--Klein (KK) spectrum of particles depends on the size and shape of the internal space $M_n$ ($n$ denoting its dimension). In a recent paper \cite{deluca-t-leaps} two of the present authors studied this relation for spin-two fields, whose masses are set by a relatively simple second-order operator on $M_n$ \cite{bachas-estes,csaki-erlich-hollowood-shirman}. Using recent mathematical results on Bakry--\'Emery geometry we established general bounds in terms of the diameter of $M_n$ or on the average of the warping function on it. 

The bounds in \cite{deluca-t-leaps} were especially useful to quantify to what extent one can achieve scale separation $m_\mathrm{KK}\gg \sqrt{|\Lambda|}$. In this paper, we focus on the opposite regime: we look at compactifications with spin-two fields of \emph{very small} mass $m_1$. This has several interesting applications, which we will review shortly.

We will see that this is achieved for example when $M_n$ is almost split in two large regions, connected by a smaller ``bridge'' (Fig.~\ref{fig:split-c}; more generally one may consider several large regions and bridges). This is inspired by a remarkable set of explicit examples in IIB string theory \cite{bachas-lavdas,bachas-lavdas2}, where it was interpreted as a ``quantum gate'', a sort of wormhole extended all along $d$-dimensional spacetime. Intuitively it is clear that such configurations should indeed lead to a small mass. Formally, the problem with two disconnected internal spaces $M^1_n$, $M^2_n$, would lead to two massless gravitons $g^1_{\mu \nu}$, $g^2_{\mu \nu}$; the bridge can be viewed as a small perturbation on this factorized problem, which then gives a small mass $m_1$ to one combination of the $g^a_{\mu \nu}$, leaving the other massless. 

\begin{figure}[ht]
\centering	
	\subfigure[\label{fig:split-c}]{\includegraphics[width=5cm]{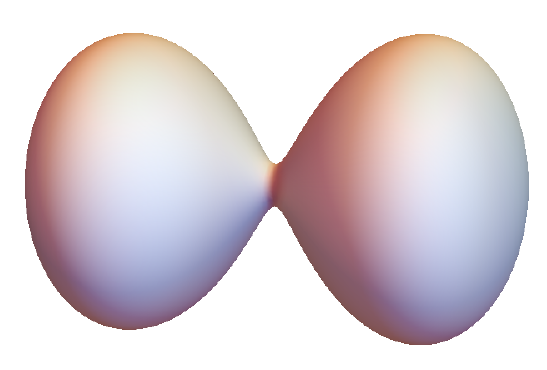}}
	\hspace{1cm}
	\subfigure[\label{fig:split-nc}]{\includegraphics[width=5cm]{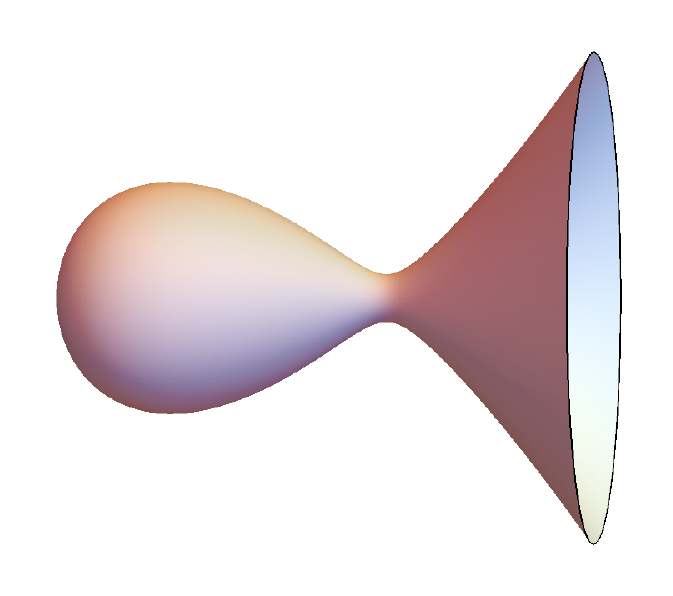}}
	\caption{\small An almost split internal space, \subref{fig:split-c} compact and \subref{fig:split-nc} non-compact.}
	\label{fig:split}
\end{figure}

More precisely, our statement is the following. We consider gravity theories in $D$ spacetime dimensions satisfying what was called Reduced Energy Condition (REC) in \cite{deluca-t-leaps}. This is the case in particular for the bulk fields in the supergravity approximation to string theory. (Localized objects can be problematic, but as we will see NS5-branes and some D-branes can be included; O-planes might also be included with some further future work). In \cite{deluca-t-leaps} it was shown that the REC implies a bound on the internal  $\bar R_{mn}- (D-2) \bar\nabla_m \partial_n A$, where the warping $A$ and the bar are defined by $\dd s^2_D = \ee^{2A}(\dd s^2_d + \dd\bar s^2_n)$. Here we use a result by two of the present authors \cite{deponti-mondino} to obtain
\begin{equation}\label{eq:intro-bound}
	 \frac14 h_1(M_n)^2 \leqslant m_1^2 \leqslant \max\left\{\frac{21}{10}h_1(M_n) \sqrt{K},\,\frac{22}{5}h_1(M_n)^2\right\}
\end{equation}
where $K= | \Lambda | + \frac{\sigma^2}{D - 2} $, and $\sigma$ is defined as the supremum of the gradient of the warping, i.e.  $(D-2)|\dd A|\leqslant \sigma $ as in \cite{deluca-t-leaps}; and 
\begin{equation}\label{eq:intro-h}
	h_1(M_n) = \mathrm{inf}_{B} \frac{\int_{\partial B} \sqrt{\bar g_{\partial B}}\,\ee^{(D-2)A}\,\dd^{n-1}x}{\int_{B} \sqrt{\bar g}\,\ee^{(D-2)A}\,\dd^{n}x}\,,
\end{equation}
where $\bar{g}_{\partial B}$ is the determinant of the pull-back metric to the boundary region $\partial B$ of $\bar{g}$, and $B$ varies among all the open sets in $M_n$ with smooth boundary, such that the denominator is smaller than $  \frac12  \int_{M_n} \sqrt{\bar g}\,\ee^{(D-2)A}\,\dd^{n}x$. This is a variant of the classical \emph{Cheeger constant}, to which it reduces for $A=0$. Sometimes, the large regions and the bridge have holographic interpretations as separate CFT$_{d-1}$ and CFT$_d$ models; in such cases, $h_1$ is precisely the ratio of their free energy coefficients ${\mathcal F}_0(\mathrm{CFT}_d)/{\mathcal F}_0(\mathrm{CFT}_{d-1})$, a measure of their degrees of freedom.

Fig.~\ref{fig:split-c} illustrates a case where $h_1$ is small, with the understanding that the size shown in the figure is according to the measure weighted by $\ee^{(D-2)A}$. The bound \eqref{eq:intro-bound} is always valid, but is most interesting for our physics application when $h_1 \ll m_D$, the Planck mass in $D$ dimensions.

Besides \eqref{eq:intro-bound}, we also consider bounds on the higher eigenvalues. For this, we consider a generalization of \eqref{eq:intro-h} involving a min-max of several Borel sets. Results of this kind were available in the Riemannian ($A=0$) case, but we generalize them to the $A\neq 0$ and non-smooth case, proving new theorems analogous to \eqref{eq:intro-bound}. In particular, as long as the lightest mass $m_1^2 \gg K$, we prove that the higher eigenvalues are bounded above by $m_k^2 < \frac{2816}5 k^2 m_1^2$; in other words, making $m_1$ small drags down all the higher eigenvalues $m_k$. This is in agreement with the \emph{spin-two swampland conjecture} of \cite{klaewer-lust-palti}. In such a regime, $m_1^2 \gg |\Lambda|$, the so-called separation of scales; this is the context where the conjecture was originally meant to apply. 
In the regime where $m_1^2 \ll|\Lambda|$, the emergence of a whole tower of light spin-two fields is also predicted by the \emph{massive-AdS-graviton conjecture} \cite[Sec.~4.1]{bachas-19}. Here we can use our results to provide examples where such a tower cannot originate from KK states; interestingly, another tower may originate from wrapped branes and save the conjecture (although we will not investigate this in the present paper).

In the particular case where $\int_{M_n} \sqrt{\bar g}\ee^{(D-2)A}\dd^{n}x$ is infinite, the $d$-dimensional Planck mass $m_d$ is infinite, and the massless graviton is non-dynamical; one is left with a single massive ``graviton'', whose mass $m_1$ is still constrained by \eqref{eq:intro-bound}; when $h_1$ is small as in Fig.~\ref{fig:split-nc}, $m_1$ is forced to be small too. Such a situation is interesting because it might be hard to tell apart experimentally from standard Einstein gravity in many respects. For this reason this scenario has been pursued energetically for a long time; see \cite{hinterbichler-massive,derham-massive} for reviews. Since we work in supergravity, the usual no-go arguments force us to focus on the AdS case; here massive gravity is both easier conceptually and of course less relevant to the real world, but the realization of small graviton mass in a UV complete theory is likely to be useful for further research.

In the afore-mentioned set of IIB examples, the lightest mass was computed in  \cite{bachas-lavdas,bachas-lavdas2}; in our language this is $m_1^2 \sim \frac34 h_1$. This agrees with our general bound \eqref{eq:intro-bound}; it is of course more precise, but our result is applicable more widely and requires almost no computation. Other famous examples of almost-split spaces can be obtained from the Maldacena--N\'u\~nez class \cite{maldacena-nunez} of solutions that include Riemann surfaces, going near the boundary of their moduli space.

To arrive at \eqref{eq:intro-bound}, as we mentioned we use a lower bound on $\bar R_{mn}- (D-2) \bar\nabla_m \partial_n A$ that was found in \cite{deluca-t-leaps}. As pointed out there, this bound is naturally interpreted in the context of so-called Bakry--\'Emery geometry \cite{bakry-emery}, which in turn implied bounds on all the spin-two masses $m_k$ in terms of $\int_{M_n}\sqrt{\bar g}\ee^{(D-2)A}\,  \dd^{n}x$ and on the diameter of $M_n$. The theorems that we are using for \eqref{eq:intro-bound}, and the new ones we will prove here,  were shown instead using the more general perspective of $\RCD$ spaces, which satisfy the so-called Riemannian--Curvature--Dimension condition. These are part of the theory of \emph{synthetic Ricci lower bounds}, which emerged recently from the encounter of ideas from optimal transport and  differential geometry  --- see for example \cite{villani-takagi, Amb} for an introduction. The word ``synthetic'' signals that this version allows for some singularities; we will show here that some brane singularities are of $\RCD$ type, with a more complete treatment deferred to \cite{deluca-deponti-mondino-t2}.

We begin in Section \ref{sec:background} by recalling some results from \cite{deluca-t-leaps}, by reviewing why singularities are useful in string theory, and with a quick review of the mathematical ideas we will need. In Section \ref{sec: D-branes RCD} we will show that some D-brane singularities satisfy the $\RCD$ condition, thus including them in the framework of our general bounds. In Section \ref{sec:bounds} we obtain our general bounds on the KK masses, including \eqref{eq:intro-h}. Sections \ref{sec:bl} and \ref{sec:riemann} are devoted to the two sets of examples we mentioned. 



\section{Background} 
\label{sec:background}

\subsection{Ricci lower bounds and Bakry--\'Emery geometry} 
\label{sub:be}

We begin here with a quick review of the relevant results from \cite[Sec.~2]{deluca-t-leaps}.

We consider a $D$-dimensional gravity theory with an Einstein--Hilbert term, with Planck mass $m_D$ and Newton's constant $\kappa^2=(2\pi)^{D-3}m_D^{2-D}$. The Einstein equations are 
\begin{equation}\label{eq:RMN}
  R_{M   N} = \frac{1}{2}\kappa^2 \left( T_{M   N} - g  _{M  
  N} \frac{T}{D-2}  \right) := \hat{T}_{M N} \, ,\qquad T_{M N} := -\frac{2}{\sqrt{-g}} \frac{\delta S_{\text{mat}}}{\delta g^{M N}}\,.
\end{equation}
Compactification vacua have the metric
\begin{equation}\label{eq:metric}
	\dd s^2_D := 
	\ee^{2 A} \dd  s^2_{d} + \dd s^2_n := \ee^{2 A} (\dd  
	s^2_{d} + \bar{\dd s}^2_n)\,,
\end{equation}
where $\dd  s^2_{d}$ is a maximally symmetric space in $d$ dimensions, and $A$ is the warping function, depending on the internal space $M_n$. 

Suppose now the stress-energy tensor satisfies the Reduced Energy Condition (REC)
\begin{equation}\label{eq:Econd}
	T^{(D)}_{m n} - \bar{g}_{m n} \frac{1}{d} T^{(d)}\geqslant 0 \,.
\end{equation}
This energy condition is satisfied for all form fields $F_{M_1\ldots M_k}$, with Lagrangian density proportional to $-\frac1{k!} F_{M_1\ldots M_k}F^{M_1\ldots M_k} := - F_k^2$, and in particular also for massless scalars. So it holds for the supergravity approximation to string theory, in $D=10$ and $11$.  The REC \eqref{eq:Econd} implies that the internal Ricci tensor obeys the bound
\begin{equation}\label{eq:RiccifBound}
  \bar{R}_{m n} - (D - 2) \bar{\nabla}_m \bar{\nabla}_n A \geqslant - \left( | \Lambda | + \frac{\sigma^2}{D - 2} 
  \right)\bar g_{mn}.
\end{equation}

Notice that \eqref{eq:RiccifBound} is strictly related with a lower bound on the so-called Bakry--\'Emery Ricci tensor for weighted Riemannian manifolds. Let us quickly recall the relevant definitions, adopting a notation that will be congenial for some of the next sections. Let $(\bar{M}, \bar{g})$ be an $n$-dimensional Riemannian manifold endowed with a weighted measure $\dd\mm:=\ee^{f} \dd{\rm vol}_{\bar{g}}$, where $\dd{\rm vol}_{\bar{g}}=\sqrt{g}\dd^n x$ denotes the standard Riemannian volume measure on $(\bar{M}, \bar{g})$. (Sometimes we will also use the notation $\dd\rm{\overline{vol}}_n$ for $\dd{\rm vol}_{\bar{g}}$, and their non-barred counterparts). The $\infty$-Bakry--\'Emery Ricci tensor $\mathsf{Ric}^{\infty,f}$ (also called more simply Bakry--\'Emery Ricci tensor and denoted as $\mathsf{Ric}^{f}$) of the weighted Riemannian manifold $(\bar{M}, \bar{g},\ee^{f} \dd\rm{ vol}_{\bar g})$ is defined as
\begin{equation}\label{eq:defBEinftyRic}
\mathsf{Ric}^{\infty,f}_x(v,v)= \mathsf{Ric}^{f}_x(v,v):=\mathsf{Ric}_x(v,v)-\mathsf{Hess}f_x(v,v), \quad \text{for all } v\in T_x\bar{M},
\end{equation}
where $\mathsf{Ric}$ denotes the standard Ricci tensor of $(\bar{M}, \bar{g})$. The components of $\mathsf{Ric}^{f}$ are exactly the left-hand side of \eqref{eq:RiccifBound}, for 
\begin{equation}\label{eq:fA}
	f= (D-2)A\,.
\end{equation}
There is also a refinement called $N$-Bakry--\'Emery Ricci tensor, for $N\in(n,\infty)$,  
defined as
\begin{equation}\label{eq:defBENRic}
\mathsf{Ric}^{N,f}_x(v,v):=\mathsf{Ric}_x(v,v)-\left[\mathsf{Hess}f+\frac{1}{N-n}  \dd f \otimes \dd f\right]_x(v,v), \qquad \textrm{if} \ N>n.
\end{equation}

The KK tower for spin-two operators is particularly simple: it is described in general \cite{csaki-erlich-hollowood-shirman,bachas-estes} by eigenfunctions $\psi_k$, $\Delta_f \psi_k = m_k^2 \psi_k$\footnote{Notice that in warped compactifications there is an ambiguity in defining the warping and the cosmological constant independently, which reverberates in the mathematical expressions for the masses. This can be understood by noticing that \eqref{eq:metric} is invariant under $A\to A+A_0$, $g_d \to e^{-2A_0} g_d$, $\bar{g_n} \to e^{-2A_0} \bar{g_n} $. This shift rescales the $m_k$ just defined, but does not affect the physical quantities, as they can be expressed as ratio of masses (such as $m_k/m_{\text{Pl}}$ or $m_k/m_\Lambda$) and thus do not rescale.} of the \emph{Bakry--\'Emery Laplacian}
\begin{equation}\label{eq:BELaplacian}
  \Delta_f (\psi) := - \frac{1}{\sqrt{\bar{g}}} \ee^{- f} \partial_m
  \left( \sqrt{\bar{g_{}}}  \bar{g}^{m n} \ee^f \partial_n \psi \right) = \Delta \psi-\bar{\nabla} f \cdot \dd \psi\,.
\end{equation}
Various results have been proven in the mathematical literature about the $m_k$, under the assumption that \eqref{eq:RiccifBound} holds. Translating these in physics terms gives some general bounds on the KK spectrum, which in particular put some constraints on scale separation, while not disallowing it. In what follows we will describe another similar bound, this time on the first non-zero eigenvalue $m_1$. This bound was found however with the help of some more sophisticated mathematics, to which we now turn.


\subsection{The $\RCD$ condition and D-branes} 
\label{sub:RCD-DBranes prel}

The bound we need has been proven by considering Bakry--\'Emery manifolds as particular cases of a more general type of spaces, the so-called $\RCD(K,N)$ spaces. This class has been introduced in \cite{AGS1} (see also \cite{G11, AGMR,  EKS, AMS, CaMi}) as a refinement of the class of $\CD(K,N)$ spaces defined and studied previously in \cite{St,St1,LoVi}. We will provide in Section \ref{sub:math prel} the formal mathematical definition of these classes, here we limit ourselves to mentioning that a $\CD(K,N)$ space is a space $X$ endowed with a distance function $\di$ and a measure $\mm$ satisfying a synthetic notion of Ricci curvature bounded from below and dimension bounded from above. The name $\CD$ stands indeed for \emph{Curvature Dimension} and the constant $K\in \R$ plays the role of the lower bound on the Ricci curvature, while $N\in (1,\infty]$ is the upper bound on the dimension. The added letter $\mathsf{R}$ in $\RCD$ stands for \emph{Riemannian}, since the main difference with respect to the $\CD$ condition is the exclusion of Finsler-like structures\footnote{A Finsler manifold has a norm on tangent spaces that is not necessarily induced by an inner product, i.e.~not necessarily quadratic.} by further require the so-called \emph{infinitesimally Hilbertianity} of the space (see Definition \ref{def:infHilb} for all the details).

One advantage of the $\RCD$ perspective that we are going to adopt is that it can also include some singularities, still offering sufficiently powerful analysis tools. We stress the fact that the space $X$ may not be a smooth Riemannian manifold, and the distance $\di$ and the measure $\mm$ are not coming in general from a Riemannian metric $g$. Indeed, besides the smooth examples of weighted Riemannian manifolds with Bakry--\'Emery Ricci curvature bounded below by $K$ \cite{St, St1}, important classes of $\mathsf{RCD}(K,N)$ spaces include suitable limits of Riemannian manifolds with Ricci curvature bounded below (the so-called Ricci limits) \cite{CC1,CC2,CC3,AGS1,GMS}, finite dimensional Alexandrov spaces \cite{Pet}, suitable stratified spaces \cite{BKMR}, appropriate quotients of Riemannian manifolds with Ricci bounded below \cite{GGKMS}.

The synthetic notion of Ricci curvature lower bound and dimension upper bound are encoded in the definition of a $\CD(K,N)$ space by imposing some convexity properties to suitable entropic functionals in the space of probability measures over $X$, taking advantage of the theory of optimal transport. A crucial fact, at the basis of the theory, is that these convexity properties are actually equivalent for a smooth Riemannian manifold to impose that the standard Ricci curvature is $\geqslant K$ and the dimension of the manifold is $\leqslant N$. Moreover, thanks to an impressive amount of work which is practically impossible to summarize here (we refer the interested reader to the survey paper \cite{Amb}), it has been shown that this definition together with the infinitesimally Hilbertianity is powerful enough to allow the study of important objects like the Laplacian and its eigenvalues, as well as to develop a first and second order calculus on these spaces and prove various functional and geometrical inequalities. 
 
The possibility to work with some non-smooth spaces is very useful for string theory, since many important compactifications have singularities induced by the back-reaction of extended objects. 
We review here the ones associated to D and M-branes, leaving a more general study for future work \cite{deluca-deponti-mondino-t2}. 

In the supergravity approximation, D$p$-branes are seen as localized objects that source gravitational and higher-form electromagnetic fields.
As for black holes in pure General Relativity or electrons in classical electrodynamics, the presence of a localized object produces a singularity in the classical fields it sources. These singularities are expected to be resolved in the full quantum theory, but are a general feature of classical limits.
In ten-dimensional supergravities, D$p$-branes are identified by a ten-dimensional metric that, in Einstein frame, asymptotes to 
\begin{equation}\label{eq:Dbrane-10}
	\dd s^2_{10} \sim H^{\frac{p-7}{8}}\left( \dd x^2_{p+1}+ H(\dd r^2+r^2 \dd s^2_{\mathbb{S}^{8-p}}) \right)\qquad \text{for}\; r\to 0 \;.
\end{equation}
Here $\dd x^2_{p+1}$ denotes the $p+1$ dimensional space parallel to the brane, (i.e.~the subspace along which the object is extended for $r\to 0$) and $r$ is a radial coordinate in the transverse directions to the object. The function $H$ is harmonic on the transverse space and it is responsible for introducing the singularity we are concerned about.
In vacuum compactifications, a D$p$-brane has to be extended along all the $d$ vacuum directions in order to preserve maximal symmetry, but in addition it can also extend among some of the internal directions. From \eqref{eq:Dbrane-10}, we obtain that the barred metric \eqref{eq:metric} approaches
\begin{equation}\label{eq:Dbrane-barred}
	\bar{\dd s}^2_n  \sim \dd x^2_{p+1-d}+H(\dd r^2+r^2 \dd s^2_{\mathbb{S}^{8-p}})\qquad \text{for}\; r\to 0 \;.
\end{equation}
Comparing with \eqref{eq:metric}, we also read that the Bakry--\'Emery function $f$ asymptotes to
\begin{equation}\label{eq:f-Dbranes}
	\ee^f= \ee^{8A} \sim H^\frac{p-7}{2}\qquad\qquad \text{for}\; r\to 0 \;.
\end{equation}
Locally, the harmonic function behaves as
\begin{equation}\label{eq:Harm}
	H\sim \left\{\begin{array}{llr}
		   (r/r_0)^{p-7}  &&1 < p < 7 \\
		  - \frac{2\pi}{g_s}\log (r/r_0) & &p = 7 
	\end{array}	\right.\qquad\qquad \text{for}\; r\to 0 \;,
\end{equation}
where $r_0^{7-p}= g_s (2\pi l_s)^{7-p}/((7-p) \mathrm{Vol}(\mathbb{S}^{8-p}))$ for $p<7$.
To analyze how these singularities affect the general results presented in Section \ref{sub:be}, we first notice that in some cases the gradient of the warping factor can be unbounded approaching the brane. Indeed, an explicit computation in the geometry \eqref{eq:Dbrane-barred} gives
\begin{equation}
  | \bar{\nabla} f |^2 = \bar{g}^{r r} \partial_r f \partial_r f = H^{- 1}
  (\partial_r f)^2 = \frac{(7 - p)^2}{4} \frac{(H')^2}{H^3}	\;.
\end{equation}
This vanishes for $p=7$ (since the warping approaches a constant) behaving in general as
\begin{equation}\label{eq:nablaF-branes}
	| \bar{\nabla} f |^2\sim  \frac{(7 - p)^4}{4} r^{5 - p} \;.
\end{equation}
\eqref{eq:nablaF-branes} is always bounded, except for D6 branes. Thus, the bound on $\mathsf{Ric}^{f}$ in \eqref{eq:RiccifBound} becomes then trivial approaching a D6-brane, since a diverging $| \bar{\nabla} f |^2$ results in an infinite $\sigma^2$. 
However, for D6-branes we can check explicitly that $\mathsf{Ric}^{N,f}$ is still bounded from below approaching the singularity, for any $N>n$, arguing as follows.

For general $\, p$,  an explicit computation in the geometry  \eqref{eq:Dbrane-barred}, with $f$ given by \eqref{eq:f-Dbranes}, results in
\begin{eqnarray}
	\frac{( \mathsf{Ric}^{N,f})_{r r}}{\bar{g}_{r r}} & = & \left( \frac{(p -
	5)}{4} - \frac{(7-p)^2}{4 (N-n)} \right) \frac{(H')^2}{H^3}\;, \\
\frac{(\mathsf{Ric}^{N,f} )_{\theta_i \theta_j}}{\bar{g}_{\theta_i 
	\theta_j}} &=&	\frac{(\mathsf{Ric})_{\theta_i \theta_j}}{\bar{g}_{\theta_i 
	\theta_j}}  =  \frac{1}{2} \frac{(H')^2}{H^3} \;.
\end{eqnarray}
The local behaviors \eqref{eq:Harm} then imply
\begin{equation}\label{eq: Ricci curvature of DBranes}
	\frac{(H')^2}{H^3}\sim\left\{\begin{array}{llr}
		    (7 - p)^2 r^{5 - p} & &1 < p < 7 \\
		  - \frac{1}{r^2 \log (r)^3}  & &p = 7
	\end{array}\right.	\qquad\qquad \text{for}\; r\to 0 \;,	
\end{equation}
from which we see that in all cases all the components of  $\mathsf{Ric}^{N,f}$ are bounded from below by some $K=K(N,n)>-\infty$, for every $N\in (n,\infty]$. 

Finally we are left with the somewhat special case of a D8-brane. The singularity introduced by this source is milder: the warping and the metric functions remain finite, with only a finite discontinuity in their first derivatives. More precisely, approaching a D8 brane at $r = 0$ the harmonic function now behaves $H\sim 1-h_8 |r|$, with $h_8$ a positive constant. Near this object the barred metric then approaches $\bar{\dd}s^2_n\sim\dd x^2_{p+1-d} + H \dd r^2 $ and the warping is given by $\ee^f \sim H^{1/2}$. In such a space both $\sigma$ and the distance from the singularity are finite, and thus the general bound \eqref{eq:RiccifBound} guarantees that $\mathsf{Ric}^{f}$ is bounded from below. This can also be checked directly: an explicit computation gives a $\mathsf{Ric}^{f} \sim h_8^2 + h_8 \delta_0 $ which is bounded from below in the distributional sense. We will be more rigorous in  Section \ref{sec: D-branes RCD}. 

To summarize, for $p\neq 6$, the general bound \eqref{eq:RiccifBound} is also valid near  D$p$-brane singularity since $\sigma^2$ stays finite, but it becomes trivial for $p = 6$. In this case, we checked explicitly that $\mathsf{Ric}^{N,f}$ is still bounded from below near a D6-brane, since a local computation shows that $\mathsf{Ric}^{N,f}$ diverges to $+\infty$.
Also, by plugging the local behavior \eqref{eq:Harm} in \eqref{eq:metric}, we see that for $p\leqslant 5$ the locus $r =0$ is at infinite distance. 

Finally, let us comment on other localized sources. First of all, fundamental strings (F1) and NS five-branes (NS5), behave exactly as D1 and D5 branes, respectively, as a consequence of the fact that our starting point, the asymptotic 10-dimensional Einstein metric \eqref{eq:Dbrane-10}, is invariant under S-duality (or more generally under the SL(2,$\mathbb{Z}$) symmetry of type IIB string theory).

For M2 and M5 branes in M-theory, the asymptotic barred metric has again the form \eqref{eq:Dbrane-barred}, now with $H\sim (r/r_0)^{q-8}$, $q = \{2, 5\}$. The warping function behaves as
\begin{equation}
	\ee^f = \ee^{9A}  \sim \left\{\begin{array}{llr}
		H^{-3}  &&\text{M2} \\
	   H^{-\frac32} & &\text{M5} 
 \end{array}	\right.\qquad\qquad \text{for}\; r\to 0 \;.
\end{equation}
In both cases $\sigma^2$ is finite, thus implying a bounded $\mathsf{Ric}^{f}$. An explicit computation shows that $\mathsf{Ric}^{N,f}$ is bounded as well. Both for M2s and M5s, the singularity is at infinite distance. 
We summarize our results in Table \ref{tab:Dp-geom}.

\begin{table}[h]
	\centering
  \begin{tabular}{|c|c|c|c|}
    \hline
    & $\mathsf{Ric}^{N,f}>K(n,N)>-\infty$ & $\sigma^2$ & distance\\
    \hline
    $0 \leqslant p \leqslant$5 & $\checkmark$ & fin. & $\infty$\\
    \hline
	$p = 6$ & $\checkmark$ & $\infty$ & fin.\\
	\hline
    $p = 7,8$ & $\checkmark$ & fin. & fin.\\
    \hline
	M2/M5  & $\checkmark$ & fin. & $\infty$\\
	\hline
  \end{tabular}
  \caption{Geometrical quantities in the transverse space for D$p$-branes and M-branes. Note $N\in (n,\infty]$. As a consequence of S-duality, F1s and NS5s behave as D1 and D5 branes respectively.}\label{tab:Dp-geom}
\end{table}

D-brane metrics are of $\RCD(K,N)$ type, as we will prove in Section \ref{sec: D-branes RCD}.
More precisely, for every $N\in (n,\infty]$ there exists $K=K(n,N)>-\infty$ such that the D-brane  is $\RCD(K,N)$ space. Notice that, outside the singularities, the D-brane is an $n$-dimensional smooth weighted Riemannian manifold with $N$-Bakry--\'Emery Ricci curvature bounded below by $K=K(n,N)>-\infty$.
As a consequence of the discussion above, these results also hold for fundamental strings, NS five branes, and M-branes in M-theory.


\subsection{Mathematical preliminaries}
\label{sub:math prel}
The aim of the following sections is to fix the notation and recall some basic constructions in the theory of metric measure spaces which play a role in the statements and the proofs of the mathematical results contained in the paper. In particular, we provide a formal definition of the $\RCD(K,N)$ class and we clarify what we mean by Laplacian, eigenvalues, Cheeger constants, in this general setting. A non-interested reader can skip these parts, the only essential fact to keep in mind is that there is a way to properly define all these notions for \emph{non-smooth spaces} in such a way that they coincide with the usual ones for (weighted) Riemannian manifolds, as discussed in the previous sections.

An oft-cited prototype for these ideas is convexity of a function $f: \mathbb{R} \to \mathbb{R}$. At the differential level this can be of course formulated as $\partial_x^2 f \geqslant 0$. But alternatively one can write it as 
\begin{equation}\label{eq:convex}
	f((1-t) x + t y) \leqslant (1-t) f(x) + t f(y) 
\end{equation}
for any $t\in [0,1]$, $x$, $y\in \mathbb{R}$. While for smooth functions these two conditions are equivalent, the ``synthetic'' \eqref{eq:convex} applies more generally. In the same way, we will describe here a version of Ricci lower bounds such as \eqref{eq:RiccifBound} that applies also to singular spaces, using optimal transport theory.

\subsubsection{Metric measure spaces}
Let $(X, \mathsf{d})$ be a complete metric space and $\mm$ a non-negative Borel measure on $X$, finite on bounded subsets. The triple $(X, \mathsf{d}, \mm)$ is called \emph{metric measure space}, m.m.s.~for short.

We call $(X,\di)$ a \emph{geodesic} metric space if every couple of points $x,y\in X$ can be joined by a geodesic in $X$, i.e.~a curve
$$\gamma:[0,1]\rightarrow X \quad  \textrm{with} \quad \di(\gamma_s,\gamma_t)=|t-s|\di(\gamma_0,\gamma_1)\quad \forall s,t\in [0,1]$$
and such that $\gamma_0=x$, $\gamma_1=y$, where we have used the notation $\gamma_t:=\gamma(t).$
The space of all geodesics in $X$ will be denoted by $\Gamma(X)$. 

Given a metric measure space $(X,\di,\mm)$, we denote by $\mathcal{P}(X)$ (resp. $\mathcal{P}{_2}(X)$) the set of Borel probability measures over $X$ (resp. the set of Borel probability measures over $X$ with finite second moment). We endow the space $\mathcal{P}{_2}(X)$  with the $2$-Kantorovich--Wasserstein distance $\mathcal{W}_2:\mathcal{P}{_2}(X)\times \mathcal{P}{_2}(X)\rightarrow [0,\infty)$, defined as

\begin{equation}\label{eq: def Wass}
\mathcal{W}^2_2(\mu_0,\mu_1):=\inf\left\{ \int_{X\times X} \di^2(x,y)\,\dd\pi(x,y) \ : \ \pi \ \textrm{coupling of} \ \mu_0 \ \textrm{and} \ \mu_1\right\}.
\end{equation}
By coupling we mean a measure $\pi\in \mathcal{P}(X\times X)$ whose \emph{marginals} $\int_X \dd y\,\pi(x,y)$ and $\int_X \dd x\,\pi(x,y)$ are respectively equal to $\mu_0(x)$ and $\mu_1(y)$. 
It can be shown that the problem \eqref{eq: def Wass} has always (at least one) minimizer $\tilde{\pi}$, called \emph{optimal coupling}.
A coupling is also called a \emph{transportation plan} since $\pi(x,y)$ describes how to move ``mass'' from $x$ into $y$, in order to transform the distribution $\mu_0$ to $\mu_1$. With this intuition, the marginality requirement on the second factor, $\int_X \dd y\,\pi(x,y)  = \mu_0(x)$, means that all the mass that is going to $y$ comes from $\mu_0$ while the other marginality condition, $\int_X \dd x\,\pi(x,y) = \mu_1(y)$, imposes that all the mass moved out from $x$ goes into $\mu_1$.
A transportation plan $\tilde{\pi}$ is then \emph{optimal} for the 2-Kantorovich--Wasserstein distance if $\tilde{\pi}$ minimizes the total cost of transforming $\mu_0$ to $\mu_1$ when the cost of moving one unit of mass from $x$ to $y$ is $\di^2(x,y)$, the square of the distance function in $X$. Thanks to this property, the distance \eqref{eq: def Wass} allows to compare two probability distributions in a way that takes into account the geometry of the underlying space.

Moreover, the optimality of $\tilde{\pi}$ is equivalent to the fact that $\tilde{\pi}$ is concentrated on a cyclical monotone set, i.e.~$\tilde{\pi}$ is an optimal coupling if and only if there exists a set $\Omega\subset X\times X$ such that $\tilde{\pi}(\Omega)=1$ and for every $k\in \N$, every permutation $\sigma$ of $\{1,\ldots, k\}$ and every $(x_1,y_1),\ldots,(x_k,y_k)\in \Omega$ it holds
$$\sum_{i=1}^k \di^2(x_i,y_i)\leqslant \sum_{i=1}^k \di^2(x_i,y_{\sigma(i)}).$$ 
This means roughly speaking that the transportation plan cannot be improved by a reshuffling of the locations to which bits of mass are transferred (see \cite[Ch.~5]{Vil}). 

The space of real-valued Lipschitz functions\footnote{Recall that a real-valued Lipschitz function $f$ over a metric space $(X,\di)$ is one for which a $K$ exists such that $|f(x)-f(y)|\leqslant K \di(x,y)$ for every $x,y\in X$.} over $X$ is denoted by $\mathsf{Lip}(X)$; we write $f\in \mathsf{Lip}_{bs}(X)$ if $f\in \mathsf{Lip}(X)$ and $f$ is bounded with bounded support. Given $f\in \mathsf{Lip}(X)$, its slope $|\nabla f|(x)$ at $x\in X$ is defined by
\begin{equation}
|\nabla f|(x):=\begin{cases}\limsup_{y\rightarrow x} \frac{|f(y)-f(x)|}{\mathsf{d}(y,x)}\quad &\textrm{if} \ x \ \textrm{is not an isolated point}, \\
0 \hspace{4cm} &\textrm{otherwise}.
\end{cases}
\end{equation}

The \emph{Cheeger energy} (see \cite{Ch99}, and \cite{AGS2} for the present formulation) is then defined as
\begin{equation}\label{def: Cheeger en}
\Ch(f)=\inf\left\{\liminf_{n\to \infty} \frac{1}{2}\int_X |\nabla f_n|^2\,\dd\mm\, :\, f_n\in \Lip(X)\cap L^2(X), f_n\rightarrow f \, \textrm{in} \ L^2(X,\mm)\right\}.
\end{equation}
In an equivalent way (see for example \cite{AGS2} or \cite[Sec.~4.4]{villani-takagi}), the Cheeger energy can be defined in terms of the \emph{minimal relaxed gradient} $|D f|$ of $f\in L^2(X,\mm)$ as
\begin{equation}
\Ch(f):=\frac{1}{2}\int_X |Df|^2\,\dd\mm ,
\end{equation}
with the convention $\Ch(f)=+\infty$ if $f$ has no relaxed gradient. In this form, it can be thought of as the generalization to possibly non-smooth spaces of the classical Dirichlet energy. Notice that the Cheeger energy is a lower semicontinuous, convex functional on $L^2(X,\mm)$.

We introduce the Sobolev space 
$$W^{1,2}(X,\di,\mm):=\{f\in L^2(X,\mm)\, : \, \Ch(f)<+\infty\},$$
endowed with the norm $\|f\|^2_{W^{1,2}}:=\|f\|^2_{L^2}+2\Ch(f).$ Notice that, at this level, the space $W^{1,2}(X,\di,\mm)$ is Banach but not in general a Hilbert space.

\begin{definition}[Infinitesimally Hilbertianity]\label{def:infHilb}
A m.m.s.~$(X,\di,\mm)$ is \emph{infinitesimally Hilbertian} if $W^{1,2}(X,\di,\mm)$ is an Hilbert space or, equivalently, if 
\begin{equation}
\Ch(f+g)+\Ch(f-g)=2\Ch(f)+2\Ch(g)\quad \textrm{for every} \ f,g\in W^{1,2}(X,\di,\mm).
\end{equation} 
\end{definition}

We finally remark that $\Lip_{bs}(X)\subset W^{1,2}(X,\di,\mm)$ since for every $f\in \Lip(X)$ it holds $|Df|(x)\leqslant |\nabla f|(x)$.

\subsubsection{Synthetic Ricci lower bounds}
An optimal coupling between two probability distributions $\mu_0$ and $\mu_1$ can be induced by a probability distribution $\nu$ in the space of geodesics. To see this, notice that for any $x\in X$, we can identify $\dd\mu_0(x)$ (resp.~$\dd\mu_1(x)$) with $\nu$ evaluated on all geodesics in $X$ that start (end) in $x$. Precisely, given a m.m.s.~$(X,\di,\mm)$, a probability measure $\nu\in \mathcal{P}(\Gamma(X))$ is called \emph{optimal dynamical plan} if $(e_0,e_1)_{\sharp}\nu$ is an optimal coupling between $(e_0)_{\sharp}\nu$ and $(e_1)_{\sharp}\nu$ for the $\mathcal{W}_2$ distance. Here, $(e_t)_{\sharp}\nu$ denotes the push forward of the measure $\nu$ through the evaluation map $e_t:\Gamma(X)\rightarrow X$ defined as $e_t(\gamma):=\gamma_t$. Intuitively, this dynamical plan tells us not only where to transport each bit of ``mass'', but also along which path.

For any geodesic space $(X,\di)$ and for any $\mu_0,\mu_1\in \mathcal{P}_2(X)$ there exists an optimal dynamical plan $\nu$ such that $(e_i)_{\sharp}\nu=\mu_i$, $i=0,1.$ If $\nu$ is an optimal dynamical plan, then $(e_t)_{\sharp}\nu$ is a Wasserstein geodesic between the marginals $(e_i)_{\sharp}\nu=\mu_i$, $i=0,1$. (We refer to \cite[Cor.~7.23]{Vil} for a proof of these facts).

Now, the behavior of geodesics is of course related to curvature. It turns out that one can use this fact to give a version of a bound on the Bakry--\'Emery curvature that is synthetic, or in other words that generalizes better to singular spaces (recall our toy example \eqref{eq:convex}). 
The details of how one gets to the correct definitions are rather involved and we cannot do them justice here, but see \cite[Ch.~16, 17]{Vil}.

Given $K\in \R, \, N\in (1,\infty), \,\theta >0$ and $t\in[0,1]$, we define the distortion coefficients as
		\begin{equation*}
		\tau^{(t)}_{K,N}(\theta) := t^{\frac{1}{N}} \sigma^{(t)}_{K, N-1}(\theta)^{\frac{N-1}{N}},
		\end{equation*}
		where 
		\begin{equation*}
		\sigma^{(t)}_{K,N}(\theta):= 
			\begin{cases}
				\infty
					&\text{if $K\theta^2\geqslant N\pi^2$},\\
				\frac{\sin(t\theta \sqrt{K/N})}{\sin(\theta \sqrt{K/N})}
					&\text{if $0<K\theta^2< N\pi^2$},\\
				t
					&\text{if $K\theta^2=0$},\\
				\frac{\sinh(t\theta \sqrt{K/N})}{\sinh(\theta \sqrt{K/N})}
					&\text{if $K\theta^2<0$}.
			\end{cases}
		\end{equation*}
		
We also introduce the \textit{relative entropy functional}:
\begin{align*}
&\mathsf{Ent}_{\mm}:\mathcal{P}_2(X)\rightarrow \R\cup \{+\infty\} \\
 \mathsf{Ent}_{\mm}(\mu):=&\begin{cases} \int \rho\log \rho \, \dd\mm \ &\textrm{if} \ \mu=\rho\,\mm \textrm{ and }  \rho\log \rho\in L^1(X,\mm), \\ +\infty \ &\textrm{otherwise},\end{cases}
\end{align*} 	
and its domain
$$
D(\mathsf{Ent}_{\mm}):=\{\mu\in \mathcal{P}_2(X) \,:\, \mathsf{Ent}_{\mm}(\mu)\in \R\}.
$$
		
We are now ready to define the curvature dimension condition.
\begin{definition}[$\mathsf{CD}(K,N)$ condition \cite{LoVi,St,St1}]
Given $K\in \R$ and $N\in(1,\infty)$, we say that $(X,\di,\mm)$ verifies the $\mathsf{CD}(K,N)$ condition if for any pair of probability measures  
	$\mu_0,\mu_1\in \mathcal{P}{_2}(X) $ with bounded support and with $\mu_0,\mu_1\ll \mm$,\footnote{Here and below we write $\mu\ll \mm$ if $\mu$ is absolutely continuous with respect to $\mm$, i.e.~for any Borel set $E$ such that $\mm(E)=0$ we have $\mu(E)=0$.} there exist a $\mathcal{W}_2$-geodesic $(\nu_t)_{t\in[0,1]}=\rho_t\mm$ with $\nu_0=\mu_0$, $\nu_1=\mu_1$ and a $\mathcal{W}_2$-optimal coupling  $\pi \in \mathcal{P}(X \times X)$
	  such that
		\begin{equation*}
		\int_X   \rho_t^{1-\frac{1}{N'}}\,\dd\mm \geqslant  \int_{X\times X}   \left[\tau^{(1-t)}_{K,N'}(\di(x,y))\rho_0^{-\frac{1}{N'}}+
		\tau^{(t)}_{K,N'}(\di(x,y))  \rho_1^{-\frac{1}{N'}} \right]  \dd \pi (x,y),
				\end{equation*}
		for any $N'\geqslant N$, $t\in[0,1]$.
		
We say that $(X,\di,\mm)$ verifies the $\mathsf{CD}(K,\infty)$ condition if for any pair of probability measures $\mu_0,\mu_1\in D(\mathsf{Ent}_{\mm})$ there exists a $\mathcal{W}_2$-geodesic $(\nu_t)_{t\in[0,1]}$ with $\nu_0=\mu_0$, $\nu_1=\mu_1$ and 
\begin{equation}
\mathsf{Ent}_{\mm}(\nu_t)\leqslant (1-t)\mathsf{Ent}_{\mm}(\nu_0)+t\mathsf{Ent}_{\mm}(\nu_1)-\frac{K}{2}t(1-t)\mathcal{W}_2^2(\nu_0,\nu_1).
\end{equation}
\end{definition}		

The $\mathsf{CD}$ condition has been later reinforced by Ambrosio, Gigli, Savar\'e \cite{AGS1} (see also \cite{G11, AGMR,  EKS, AMS, CaMi}), who introduced the so-called $\mathsf{RCD}$ condition.
		
\begin{definition}[$\mathsf{RCD}(K,N)$ condition]\label{def:RCD}
Given $K\in \R$ and $N\in(1,\infty]$, we say that $(X,\di,\mm)$ verifies the $\RCD(K,N)$ condition if $(X,\di,\mm)$ satisfies the $\CD(K,N)$ condition and it is infinitesimally Hilbertian.
\end{definition}

\begin{remark}
	Every $n$-dimensional Riemannian manifold $(M,g)$ equipped with the geodesic distance $\di$ and the weighted measure $\dd\mm(x)=\ee^{-f(x)}\dd\mathsf{vol}_g$ is $\RCD(K,N)$ if and only if $N\in [n,\infty]$ and 
$$\mathsf{Ric}^{N,f}_x(v,v)\geqslant K\|v\|^2_{T_xM}, \qquad \textrm{for every} \ v\in T_xM,$$
where $\mathsf{Ric}^{n,f}:= \mathsf{Ric}$ and $\mathsf{Ric}^{N,f}$ was defined in \eqref{eq:defBEinftyRic} and \eqref{eq:defBENRic}.
\end{remark}
With respect to the $\CD$ condition, the $\RCD$ condition rules out the Finsler manifolds that are not Riemannian. 

It is not difficult to see that $\RCD(K,N)\subset \RCD(K,\infty)$, and the inclusion is strict as shown for instance by considering the Gaussian space $\big(\R^n,|\cdot|, \ee^{-\frac{|x|^2}{2}}\big)$ which is $\RCD(1,\infty)$ but not $\RCD(1,N)$ for any finite $N$.

It follows from the definition of $\RCD$ space that the Banach space $W^{1,2}(X,\di,\mm)$ is actually Hilbert, and one can prove that $\Lip_{bs}(X)\subset W^{1,2}(X,\di,\mm)$ with dense inclusion.

\subsubsection{Laplacian, spectrum and Cheeger constants}
In this section we assume $(X,\di,\mm)$ to be a $\RCD(K,\infty)$ space, $K\in \R$. We define the set $D(\Delta)$ as the set of $f\in L^2(X,\mm)$ such that $\partial \Ch(f)\neq \emptyset$, where $\partial \Ch$ denotes the subdifferential of the Cheeger energy.\footnote{Given $f\in W^{1,2}(X)$, we say $g\in\partial \Ch(f)$ if $\int_Xg(\psi-f)\,\dd\mm\le \Ch(\psi)-\Ch(f) \quad \forall \psi\in L^2(X).$} For $f\in D(\Delta)$ we define $\Delta f$ as the element of minimal $L^2(X,\mm)$ norm in $\partial \Ch(f)$. Notice that we are adopting here the convention that $\Delta$ has a non-negative spectrum, according to the physical literature.

The operator $\Delta$ is a non-negative, densely defined, self-adjoint operator on the Hilbert space $L^2(X,\mm)$. A number $\lambda\in \mathbb{C}$ is a \emph{regular value} of $\Delta$ if $(\lambda \rm{Id}+\Delta)$ has a bounded inverse. The \emph{spectrum} of $\Delta$ is the set $\sigma(\Delta)$ of numbers $\lambda\in [0,\infty)$ that are not regular values. We call a non-zero function $f\in D(\Delta)$ an \emph{eigenfunction} of $\Delta$ of \emph{eigenvalue} $\lambda$ if it holds $\Delta f=\lambda f$. The set of all eigenvalues forms the so-called \emph{point spectrum}. The \emph{discrete spectrum} $\sigma_d(\Delta)$ is the set of all eigenvalues that are isolated in the point spectrum with finite dimensional eigenspace, and finally the \emph{essential spectrum} can be defined as $\sigma_{ess}(\Delta):=\sigma(\Delta)\setminus \sigma_{d}(\Delta)$. 
We also denote by $\Sigma$ the infimum of the essential spectrum of $\Delta$, i.e.
$$\Sigma:=\inf \sigma_{ess}(\Delta) \qquad \textrm{and} \qquad \Sigma:=+\infty \ \textrm{if} \ \sigma_{ess}(\Delta)=\emptyset.$$
We remark that $\Sigma=+\infty$ for every bounded $\RCD(K,\infty)$ space, since it remains true even in this general setting that spaces with finite diameter have discrete spectrum \cite{GMS}.

Given $f\in W^{1,2}(X,\di,\mm)$, $f\not\equiv 0$, its Rayleigh quotient is defined as
\begin{equation}\label{eq: Rayleigh} 
\mathcal{R}(f):=\frac{2\Ch(f)}{\int_X |f|^2\,\dd\mm}\, .
\end{equation}

It is well known that the eigenvalues of $\Delta$ can be characterized variationally. More precisely, the set of eigenvalues below $\Sigma$ is at most countable and, listing them in an increasing order $\lambda_0<\lambda_1\leqslant ...\leqslant \lambda_k\leqslant ...$, it holds
\begin{equation}\label{eq:defeigk}
\lambda_k= \min_{V_{k+1}}\, \max_{f\in V_{k+1},\\ f\not\equiv 0} \ \mathcal{R}(f)\, ,
\end{equation}
where $V_{k}$ denotes a $k$-dimensional subspace of $W^{1,2}(X,\di,\mm)$ (see for instance \cite{Dav} for a proof of this general version of the min-max principle). 

Given a Borel subset $B\subset X$ with $\mm(B)<\infty$, the \textit{perimeter} $\mathrm{Per}(B)$ is defined as follows:
\begin{equation*}
\mathrm{Per}(B):=\inf\bigg\{\liminf_{n\rightarrow \infty}\int_X |\nabla f_n|\,\dd\mm: f_n\in \mathsf{Lip}_{bs}(X), f_n\rightarrow \chi_B \ \mathrm{in} \ L^1(X,\mm)\bigg\}.
\end{equation*}
This seemingly complicated definition reduces to $\int_{\partial B} \ee^f \dd\rm{vol}_{n-1}$ for a measure $\dd\mm=\ee^{f} \dd{\rm vol}_g$ such as those discussed below \eqref{eq:RiccifBound}, and for a set $B$ with a smooth boundary $\partial B$.

We define the $k$-\emph{Cheeger constant} (or \emph{$k$-way isoperimetric constant}) as
\begin{equation}\label{def: m-Cheeger}
h_k(X):=\inf_{B_0,..,B_k}\, \max_{0\leqslant i\leqslant k} \frac{\mathrm{Per}(B_i)}{\mm(B_i)}\, ,
\end{equation}
where the infimum runs over all collections of $k+1$ disjoint, Borel sets $B_i\subset X$ such that $0<\mm(B_i)<\infty$ (see \cite{lee-gharan-trevisan, funano, Miclo, Liu}).
It is easy to see that $h_k(X)\leqslant h_{k+1}(X)$ for every $k\in \N$ and, for spaces with finite measure, $h_0(X)=0$. 

The case $k=1$ plays a prominent role. Indeed, $h_1$ corresponds to the classical isoperimetric constant introduced by Cheeger to bound from below the first eigenvalue of the Laplacian on a compact Riemannian manifold \cite{cheeger-bound}, justifying the name of these constants. 
Notice that for spaces of finite measure
\begin{equation}\label{eq:defChConst}
h_1(X)=\inf  \left\{\frac{\mathrm{Per}(B)}{\mm(B)}\, :\, B\subset X \text{ Borel subset with $0<\mm(B)\leqslant \mm(X)/2$} \right\}. 
\end{equation}
For a measure $\dd\mm=\ee^{f} \dd{\rm vol}_g$, with $f=(D-2)A$ as in \eqref{eq:fA}, this reduces to \eqref{eq:intro-h}  thanks to the well known approximation of a finite perimeter set by sets with smooth boundary (see for instance \cite[Theorem 3.42]{AFP} for a proof in the Euclidean framework, which can be easily adapted to the smooth weighted Riemannian setting by working in local coordinates).

\section{D-brane type singularities are $\RCD$}
\label{sec: D-branes RCD}

The aim of this section is to prove that the singular metrics and the corresponding weighted measures appearing in the definition of a D-brane satisfy the $\RCD$ condition. 

To be more precise, for $p \leqslant 5$ we will see that there is essentially nothing to prove; for D6- and D7-branes we need separate and lengthy treatments, which will however apply to the physically unrealistic situation where  the metrics and the weighted measure are \emph{exactly} of the form \eqref{eq:Dbrane-barred} and \eqref{eq:f-Dbranes} respectively near the tips of the ends. In concrete solutions, this will only be the case \emph{asymptotically} as $r\to 0$. Encouraged by the current exact results, we expect that a rigorous result can be proved also for asymptotic D6- and D7-brane metric adapting the techniques in \cite{BKMR}; we will investigate this in future work \cite{deluca-deponti-mondino-t2}. 
These results also hold for fundamental strings and NS five-branes, which behave exactly as D1 and D5 branes, and for M branes as well, since they enjoy the same curvature properties of the $p<5$ case (cf.~Table \ref{tab:Dp-geom}).

In order to properly formulate the statement of our theorem, let us define the metric measure structure that we are going to consider.

\begin{definition}[D-brane type metric measure spaces]
We define a \emph{D-brane type metric measure space} as a smooth and compact Riemannian manifold $(X,g)$ glued (in a smooth way) with a finite number of ends where the metric $g$ is of the form \eqref{eq:Dbrane-barred}
in a neighborhood of the point $\{r=0\}$.
Recall that $\dd x_{p+1-d}^2$ is the flat metric of the $(p+1-d)$-dimensional Euclidean space, $\dd s^2_{\mathbb{S}^{8-p}}$ is the round metric on the $(8-p)$-dimensional sphere $\mathbb{S}^{8-p}$ and $H(r)$ is defined as in 
\eqref{eq:Harm}.
We endow $X$ with a weighted measure and view it as a metric measure space $(X,\di,\mm)$ where:
\begin{itemize}
\item the distance $\di$ between two points $p,q \in X$ is given by 
$$\di(p,q):=\inf_{\gamma\in \Gamma(p,q)} \int g\left(\gamma'(t),\gamma'(t)\right) \dd t,$$
with $\Gamma(p,q)$ denoting the set of absolutely continuous curves joining $p$ to $q$.
\item the measure $\mm$ is a weighted volume measure $\mm:=\ee^f\mathsf{\dd vol}_g$, with $\ee^f$ smooth outside the tips of the ends and 
equal to \eqref{eq:f-Dbranes}
in a neighborhood of the point $\{r=0\}$.
\end{itemize}
\end{definition}

We are now ready to state our main result of this section.
\begin{theorem}\label{th: D-branes are RCD}
For every $n$-dimensional D-brane type metric measure space $(X,\sfd, \mm)$ and for every $N\in (n,\infty]$ there exists $K=K(n,N)>-\infty$ such that $(X,\sfd, \mm)$  is an $\RCD(K,N)$ space.
\end{theorem}

\begin{proof}[Proof of Theorem \ref{th: D-branes are RCD}]
Notice that it is not restrictive to assume that $X$ has only one end. Moreover, in our argument we can also neglect the flat part given by the Euclidean metric $\dd x^2_{p+1-d}$ thanks to the tensorization property of $\RCD$ spaces (see  \cite[Theorem 7.6]{AGMR} after \cite[Theorem 6.13]{AGS1} for $N=\infty$ and \cite[Theorem 3.23]{EKS} for $N\in (1,\infty)$).

We denote by $\mathsf{O}:=\{r=0\}$ the tip of the end. As we are going to explain, the point $\mathsf{O}$ may or may not be included in $X$.

Since the distance and the measure are smooth in the compact region outside a neighbourhood of $\mathsf{O}$, it is clear that we only have to check that the $\RCD(K,N)$ condition is satisfied for a space $X$ with metric
$$\bar{g}=H(r)(\dd r^2+r^2 \dd s^2_{\mathbb{S}^{8-p}})$$
and measure 
$$\bar{\mm}=H(r)^{\frac{p-7}{2}}\dd\mathsf{vol}_{\bar{g}}$$
near $\mathsf{O}$. 
\medskip

With this notation and these assumptions in mind, we divide the proof in three cases and we start to prove that the $\CD(K,N)$ condition is satisfied:
\begin{itemize}
\item \textbf{Case $p\leqslant 5$.}
We notice that the distance between $\mathsf{O}$ and any other point is infinite. We thus do not include the point $\mathsf{O}$ in $X$ so that the space $(X, \di)$ is a complete metric space. 
Since the Ricci curvature stays bounded from below when $r\to 0$ as shown in \eqref{eq: Ricci curvature of DBranes}, we can conclude that $(X,\di,\mm)$ is a $\CD(K,N)$ space (actually, a smooth manifold).

\item \textbf{Case $p=6$.} 
In this case we consider $\mathsf{O}\in X$. To prove that $(X,\di,\mm)$ is an $\CD(K,N)$ space we adapt the strategy proposed by Bacher and Sturm in \cite{BaSt}.
Indeed, after the change of variable $\rho=2\sqrt{r}$ the metric $\bar{g}$ takes the form
\begin{equation}\label{eq: cone metric2}
\bar{g}=\dd\rho^2+\frac{1}{4}\rho^2 \dd s^2_{\mathbb{S}^2},
\end{equation}
with measure 
\begin{equation}\label{eq: cone weighted meas2}
\bar{\mm}=\frac{1}{8}\rho^3\,\dd\mathsf{vol}_{\rho}\,\dd\mathsf{vol}_{\mathbb{S}^2} 
\end{equation}
in a neighbourhood of the point $\mathsf{O}$.

Notice that from the metric point of view we can infer that our manifold is locally a cone. Moreover $\bar{\mm}$ is absolutely continuous with respect to the standard cone measure and thus we can follow verbatim \cite[Theorem 4]{BaSt} and show that any optimal dynamical plan with first marginal absolutely continuous with respect to $\bar{\mm}$ gives zero mass to geodesics through $\mathsf{O}$. 
By \eqref{eq: Ricci curvature of DBranes} we also know that the weighted Ricci curvature stays bounded from below by a constant on the set $\{\rho>0\}$. 

Thus, as in \cite[Theorem 5]{BaSt} and using the characterization of the curvature dimension condition given in \cite[Lemma 2]{BaSt}, we can deduce that the $\CD$ condition is satisfied since this is true on the set $\{\rho>0\}$ and almost any (with respect to any optimal dynamical plan $\nu$ such that $(e_{0})_{\sharp} \nu \ll \bar{\mm}$) $\gamma\in \Gamma(X)$ stays on this set.

\item\textbf{Case $p=7$.} 
We consider $\mathsf{O}\in X$. We make the change of variable $\rho=\int_0^r \sqrt{-\log(s)}\,\dd s$ so that the metric $\bar{g}$ takes the form
\begin{equation}\label{eq: cone metric1}
\bar{g}=\dd\rho^2+f(\rho)\dd s^2_{\mathbb{S}^1},
\end{equation}
with measure 
\begin{equation}\label{eq: cone weighted meas1}
\bar{\mm}=\sqrt{f(\rho)}\,\dd\mathsf{vol}_{\rho}\,\dd\mathsf{vol}_{\mathbb{S}^1}
\end{equation}
in a neighbourhood of the point $\mathsf{O}$, where $f(\rho)$ corresponds to the factor $-\log(r)r^2$ written in the new variable $\rho$. Since 
\begin{equation}\label{eq: comparison after changeofvariable1}
\sqrt{-\log(r)r^2}\leqslant \int_0^r \sqrt{-\log(s)}\,\dd s \qquad \textrm{for} \ 0<r<1,
\end{equation}
as one can easily notice by comparing the derivative of the two sides of \eqref{eq: comparison after changeofvariable1}, we have 
$$0\leqslant f(\rho)\leqslant \rho^2.$$
This fact and the expression of $\bar{g}$ lead to the following two crucial facts for the distance $\di_{\bar{g}}$ associated to $\bar{g}$:
\begin{itemize}
\item For every $x\in \mathbb{S}^1$ we have $\di_{\bar{g}}((x,\rho),\mathsf{O})=\rho$.
\item the distance $\di_{\bar{g}}$ is bounded from above by the standard cone distance, i.e.~for every couple of points $A:=(x_0,\rho_0),B:=(x_1,\rho_1)$ we have  
$$\di_{\bar{g}}(A,B)\leqslant \di_{C(\mathbb{S}^1)}(A,B):=\sqrt{\rho_0^2+\rho_1^2-2\rho_0\rho_1\cos(\di_{\mathbb{S}^1}(x_0,x_1)}).$$ 
\end{itemize} 
Again, to prove that $X$ endowed with the given distance and measure is $\CD(K,N)$ we adapt the strategy proposed by Bacher and Sturm \cite{BaSt} to which we refer for all the details. First of all, we notice that thanks to \eqref{eq: Ricci curvature of DBranes} the weighted Ricci curvature stays bounded from below by a constant on the set $\{\rho>0\}$.

Thanks to the argument of \cite{BaSt}, we know that the result follows if we prove that 
\begin{enumerate}[(a)]
\item for any optimal dynamical plan $\nu$ such that $(e_{0})_{\sharp} \nu \ll \bar{\mm}$ we have $\nu(\Gamma_{\mathsf{O}})=0$, where
\begin{equation}\label{eq:defGammaO}
\Gamma_{\mathsf{O}}:=\{\gamma\in \Gamma(X): \gamma_t=\mathsf{O} \ \textrm{for some} \ t\in(0,1)\},
\end{equation}
\end{enumerate}
and the well-posedness and the proof of the statement (a) is a consequence of the following two facts (see \cite[Theorem 4]{BaSt}):
\begin{enumerate}[(1)]
\item For every $t\in (0,1)$ there exists at most one geodesic $\gamma\in \supp(\nu)$ with $\gamma_t=\mathsf{O}$.
\item For every $\rho>0$ there exists at most one $x\in \mathbb{S}^1$ such that $\gamma_0=(x,\rho)$ is the initial point of some geodesic $\gamma\in\supp(\nu)\cap \Gamma_{\mathsf{O}}$.
\end{enumerate}

Indeed, once we have $(1)$ and then $(2)$, (a) can be proved by contradiction by assuming that $\nu$ is supported on $\Gamma_{\mathsf{O}}$ with $\nu(\Gamma_{\mathsf{O}})>0$. We can also assume that $\gamma_0\neq \mathsf{O}$, $\gamma_1\neq \mathsf{O}$ for $\nu$-a.e.~$\gamma$ (or in other words, that the set of geodesics that neither start from nor end in $\mathsf{O}$ has measure zero with respect to $\nu$). This can be done since $\bar{\mm}$ gives no mass to $\mathsf{O}$. The desired contradiction can now be reached since by $(2)$ the measure $\mu_0:=(e_0)_{\sharp}(\nu)$ is concentrated on a set of the form 
$$C_f:=\{(f(\rho),\rho):\rho>0\}$$
for some function $f$, and $\bar{\mm}(C_f)=0$. 

Thus, it remains to prove $(1)$ and $(2)$. 
\\
To prove $(1)$, let us consider two geodesics $\gamma,\gamma'\in\supp(\nu)$ passing at time $t$ through $\mathsf{O}$. We have $\gamma_0=(x_0,t\rho)$, $\gamma_1=(x_1,(1-t)\rho)$ for some $x_0,x_1\in\mathbb{S}^1$ and similarly  $\gamma'_0=(x'_0,t\rho')$, $\gamma'_1=(x'_1,(1-t)\rho')$. If $\rho>0$ (resp.~$\rho'>0$), we can assume that $x_0,x_1$ (resp.~$x'_0,x'_1$) are antipodal. This is a consequence of the following Lemma:
\begin{lemma}
Let $\gamma:[0,1]\rightarrow X$ be a non-constant geodesic with endpoints $\gamma_0=(x_0,\rho_0)$ and $\gamma_1=(x_1,\rho_1)$. If $\gamma_t=\mathsf{O}$ for some $t\in (0,1)$, then $x_0$ and $x_1$ are antipodal as points in $\mathbb{S}^1$.
\end{lemma}
\begin{proof}
Due to the expression of $\di_{\bar{g}}$, we know that $\rho_0=t\di_{\bar{g}}(\gamma_0,\gamma_1)$ and $\rho_1=(1-t)\di_{\bar{g}}(\gamma_0,\gamma_1)$ which implies $\rho_1=\frac{1-t}{t}\rho_0$.
In particular
\begin{equation}
\frac{\rho_0^2}{t^2}=\di^2_{\bar{g}}(\gamma_0,\gamma_1)\leqslant \di^2_{C(\mathbb{S}^1)}(\gamma_0,\gamma_1)=\rho_0^2+\frac{(1-t)^2}{t^2}\rho_0^2-2\frac{(1-t)}{t}\rho^2_0\cos(\di_{\mathbb{S}^1}(x_0,x_1))
\end{equation}
from which $\cos(\di_{\mathbb{S}^1}(x_0,x_1))\leqslant -1$ that proves the claim.
\end{proof}
Now the proof of the point $(1)$ is a consequence of the cyclical monotonicity. Indeed, using the triangle inequality for $\di_{C(\mathbb{S}^1)}$, we know that
\begin{align*}
0&\leqslant \di^2_{\bar{g}}(\gamma_0,\gamma'_1)+\di^2_{\bar{g}}(\gamma'_0,\gamma_1)-\di^2_{\bar{g}}(\gamma_0,\gamma_1)-\di^2_{\bar{g}}(\gamma'_0,\gamma'_1)\\
&\leqslant \di^2_{C(\mathbb{S}^1)}(\gamma_0,\gamma'_1)+\di^2_{C(\mathbb{S}^1)}(\gamma'_0,\gamma_1)-\rho^2-\rho'^{\,2}\\
&\leqslant [t\rho+(1-t)\rho']^2+[t\rho'+(1-t)\rho]^2\rho^2-\rho'^{\,2}=-2t(1-t)(\rho-\rho')^2
\end{align*}
which implies $\rho=\rho'$. Using now this information, we can derive
\begin{align*}
0&\leqslant \di^2_{\bar{g}}(\gamma_0,\gamma'_1)+\di^2_{\bar{g}}(\gamma'_0,\gamma_1)-\di^2_{\bar{g}}(\gamma_0,\gamma_1)-\di^2_{\bar{g}}(\gamma'_0,\gamma'_1)\\
&\leqslant \di^2_{C(\mathbb{S}^1)}(\gamma_0,\gamma'_1)+\di^2_{C(\mathbb{S}^1)}(\gamma'_0,\gamma_1)-2\rho^2\\
&=-2\rho^2t(1-t)[2+\cos(\di_{\mathbb{S}^1}(x_0,x'_1))+\cos(\di_{\mathbb{S}^1}(x'_0,x_1))]
\end{align*}
and in particular $x_0$ and $x'_1$ are antipodal and thus $x_0=x'_0$ and $x_1=x'_1$. 

To prove $(2)$ we consider $\gamma,\gamma'\in\supp(\nu)\cap \Gamma_{\mathsf{O}}$ with $\gamma_0=(x_0,\rho)$, $\gamma'_0=(x'_0,\rho)$ where $\rho>0$. Since by assumption $\gamma$ and $\gamma'$ are geodesics that pass through the origin, we also know that $\gamma_1=(x_1,\rho_1)$, $\gamma'_0=(x'_1,\rho'_1)$ where $\di_{\mathbb{S}^1}(x_0,x_1)=\pi$, $\di_{\mathbb{S}^1}(x'_0,x'_1)=\pi$ and $\rho_1$, $\rho'_1$ are positive real numbers.
Again by cyclical monotonicity and the upper bound on $\di_{\bar{g}}$ we have 
\begin{align*}
0&\leqslant \di^2_{\bar{g}}(\gamma_0,\gamma'_1)+\di^2_{\bar{g}}(\gamma'_0,\gamma_1)-\di^2_{\bar{g}}(\gamma_0,\gamma_1)-\di^2_{\bar{g}}(\gamma'_0,\gamma'_1)\\
&\leqslant \di^2_{C(\mathbb{S}^1)}(\gamma_0,\gamma'_1)+\di^2_{C(\mathbb{S}^1)}(\gamma'_0,\gamma_1)-(\rho+\rho_1)^2-(\rho+\rho'_1)^2\\
&=-2\rho\rho'_1[1+\cos(\di_{\mathbb{S}^1}(x_0,x'_1)]-2\rho\rho_1[1+\cos(\di_{\mathbb{S}^1}(x'_0,x_1)]
\end{align*}
that force $x_0$ and $x'_1$ to be antipodal and thus $x_0=x'_0$.

\item\textbf{Case $p=8$.}
Let us first recall the expression of the metric and measure of an $8$-brane. Let  $h_{8}>0$ be a positive constant and set
\begin{equation}\label{eq:defH}
H(r):= 1-h_{8} |r|, \quad \text{for } r\in [- (2h_{8})^{-1}, (2h_{8})^{-1}]. 
\end{equation}
The metric and the measure are given respectively by
\begin{equation}\label{eq:mmD8}
\bar{g}=  H(r)\, \dd r^2, \quad  \bar{\mm}=\ee^{f} \dd \text{vol}_{\bar{g}}= H^{1/2} \dd \text{vol}_{\bar{g}}, \quad   \text{for } r\in [- (2h_{8})^{-1}, (2h_{8})^{-1}]. 
\end{equation} 
As above, by the tensorization property, it is enough to check that the 1-dimensional metric measure space
\begin{equation}\label{eq:mmD81D}
 \Big([- (2h_{8})^{-1}, (2h_{8})^{-1}],  H(r)\, \dd r^2, H^{1/2} \dd \text{vol}_{\R} \Big)
 \end{equation}
 satisfies $\CD(K,N)$ for any $N>1$, for  some $K=K(N)\in \R$.
Notice that  the $C^{1}$-diffeomorphism
\begin{equation}
t(r):=- \frac{2}{3h_{8}} \sgn(r) \left[ (1-h_{8} |r|)^{3/2} -1 \right]\, 
\end{equation}
gives an isomorphism of metric measure spaces between the orginal \eqref{eq:mmD81D}
 with the weighted Euclidean segment 
\begin{equation}\label{eq:mmD81DNew}
 \left( \left[ - \frac{2}{3h_{8}} (1-  (2)^{-3/2}),  \frac{2}{3h_{8}} (1-  (2)^{-3/2})  \right], \sfd_{eucl}, \left( 1- \frac{3}{2} h_{8} |t| \right)^{1/3} \dd \text{vol}_{\R}  \right).
 \end{equation}
 Hence, our claim is equivalent to check that the m.m.s. in \eqref{eq:mmD81DNew} satisfies $\CD(K,N)$.
Recall that a segment $(I, \sfd_{eucl}, \ee^{f} \, \dd \text{vol}_{\R})$ satisfies $\CD(K,N)$ if and only if $f$ is semi-concave (thus, in particular, locally Lipschitz and twice differentiable a.e.) and it satisfies 
\begin{equation}\label{eq:CDKNCond1D}
f''+\frac{1}{N-1}(f')^{2}\leq -K
\end{equation}
in distributional sense (see for instance \cite[Appendix A]{CaMi}). Now, a direct computation gives that the left hand side of \eqref{eq:CDKNCond1D} corresponding to the space \eqref{eq:mmD81DNew} is
$$
\frac{3}{4} h_{8}^{2} \left( 1- \frac{3}{2} h_{8} |t| \right)^{-2} \left(\frac{1}{N-1}-1 \right) - \frac{1}{2} h_{8} \delta_{0}<0 
$$
in distributional sense on  $\left[ - \frac{2}{3h_{8}} (1-  (2)^{-3/2}),  \frac{2}{3h_{8}} (1-  (2)^{-3/2})  \right]$, where $\delta_0$ denotes the Dirac mass distribution centred at $0\in \R$. We conclude that the space \eqref{eq:mmD81D}  satisfies $\CD(0,N)$ for every $N>1$, and thus the metric measure space associated to the (barred, i.e. transverse part of the)  8-Brane \eqref{eq:mmD8} satisfies $\CD(0,N)$, for every $N>1$.

\end{itemize}

This concludes the proof that $(X, \sfd, \mm)$ is a $\CD(K,N)$ space. In order to show that it is an $\RCD(K,N)$ space, we need to show that the Cheeger energy is a quadratic form.
As above, the only possible non-trivial cases are $p=6,7$. Recall that, in both cases, we proved that for any optimal dynamical plan $\nu$ such that $(e_{0})_{\sharp} \nu \ll \bar{\mm}$ we have $\nu(\Gamma_{\mathsf{O}})=0$, where $\Gamma_{\mathsf{O}}$ is the set of geodesics passing through the singular point $\mathsf{O}$, see \eqref{eq:defGammaO}.
By the approach to weak upper gradients via optimal transport (see \cite{AGS2}), it follows that the weak upper gradient    coincides with the modulus of the standard Riemannian gradient  on the smooth part $X\setminus\{\mathsf{O}\}$.
The infinitesimal Hilbertianity now easily follows by using that $\mm(\{\mathsf{O}\})=0$:
\begin{align*}
\Ch(f+g)+\Ch(f-g)&=\int_X |D(f+g)|^2\,\dd\mm + \int_X |D(f-g)|^2\,\dd\mm  \\
&= \int_{X\setminus\{\mathsf{O}\}} |D(f+g)|^2\,\dd\mm + \int_{X\setminus\{\mathsf{O}\}} |D(f-g)|^2\,\dd\mm  \\
&= 2  \int_{X\setminus\{\mathsf{O}\}} |Df|^2\,\dd\mm + 2 \int_{X\setminus\{\mathsf{O}\}} |Dg|^2\,\dd\mm \\
&= 2  \int_{X} |Df|^2\,\dd\mm + 2 \int_{X} |Dg|^2\,\dd\mm \\
&=2\Ch(f)+2\Ch(g).
\end{align*}

\end{proof}

\section{New bounds on spin-two masses} 
\label{sec:bounds}

As we have described in Section \ref{sub:be}, the masses of the spin-two fluctuations around any vacuum compactification are given by the eigenvalues of a certain universal operator \cite{csaki-erlich-hollowood-shirman,bachas-estes}, which in the language of \emph{Bakry--\'Emery geometry} can be rewritten as the weighted Laplacian \eqref{eq:BELaplacian}. In \cite{deluca-t-leaps}, this connection to Bakry--\'Emery geometry was exploited to put rigorous bounds on the spin-two masses in terms of either the internal diameter or the reduced Planck mass. The results followed from the curvature bound \eqref{eq:RiccifBound} and from the application of known mathematical theorems that bound the eigenvalues of the Bakry--\'Emery Laplacian in term of the weighted volume or the diameter.

In this section we present new bounds on the masses of spin-two fields in gravity compactifications  in terms of the isoperimetric constants. We discuss the first spin-two state in Section \ref{sub:first} and the higher ones in Section \ref{sub:higher}. In Section \ref{sub:hol} we then interpret them from the holographic point of view. 
These bounds are obtained by applying both known and new mathematical results, which we present and prove in Section \ref{sub:proofs}. As we will see, these new results hold in the more general $\RCD$ setting introduced in Section \ref{sub:math prel} and thus apply to general compactifications with brane sources, as discussed in Section \ref{sec: D-branes RCD}.

\subsection{The first eigenvalue} 
\label{sub:first}

In the classical compact Riemannian setting, the measure is just $\dd\mm=\ee^{f} \dd{\rm vol}_g$, and the first Cheeger constant $h_1$ introduced in \eqref{eq:defChConst} is simply
\begin{equation}\label{eq:h1}
	h_1 =  \inf_{B} \frac{\text{Vol}(\partial B)}{\text{Vol}(B)}\;;
\end{equation}
in this situation the infimum can be taken among all the $n$-dimensional smooth sets $B\subset M_n$ such that $0<\text{Vol}(B)\leqslant \frac{1}{2} \text{Vol}(M_n)$.
Without any assumption on the curvature of $M_n$, the \emph{Cheeger inequality} \cite{cheeger-bound} provides a lower bound on the first positive eigenvalue of the standard Laplacian:
\begin{equation}\label{eq:cheeger0}
\lambda_1\geqslant \frac{1}{4} h_1^2 \,.
\end{equation}

If the Ricci curvature is bounded from below by a non-positive constant $K$, $\text{Ric} \geqslant K$, Buser \cite{buser-cheeger} also found an upper bound in terms of $h_1$ and $K$:
\begin{equation}\label{eq:buser0}
	\lambda_1 \leqslant 2 h_1 \sqrt{-(n-1)K}+ 10 h_1^2\;.
\end{equation}
If $\text{Vol}(M_n)$ is infinite, the constant functions are not square integrable and thus $\lambda_0$ may be strictly positive. In this case, a version of Cheeger's and Buser's inequalities still holds by replacing $\lambda_1$ with $\lambda_0$ and $h_1$ with $h_0$, which is defined again as in \eqref{eq:h1} but with the infimum now taken among all the subsets of $M_n$ of finite volume. 

These classical geometrical results are encouraging, since they allow to constrain $\lambda_1$ pretty tightly for manifolds like the ones in Figure \ref{fig:split}, which have a small $h_1$ (or $h_0$) due to their bottlenecks. However, they are of limited use for general gravity compactifications, which generically have a non-trivial warping factor $A$ as in \eqref{eq:metric}.
In this case, the masses of the spin-two fields are not given by the eigenvalues of a standard Laplacian, and the equations of motion do not allow to extract a general lower bound on the Ricci tensor, since terms involving second derivatives of $A$ do not have a definite sign. 
Moreover, interesting dynamics often requires the presence of localized sources which, in addition to sourcing a non-trivial warping factor, can introduce singularities in the internal space, forbidding a direct application of the classical results obtained in the smooth setting.

Luckily, these difficulties can be overcome with a natural extension of the Riemannian setting. As shown in \cite{deluca-t-leaps}, in compactifications of theories that satisfy the REC \eqref{eq:Econd}, the Bakry--\'Emery generalization of the Ricci curvature tensor is bounded from below by a negative constant \eqref{eq:RiccifBound} only depending on the cosmological constant and on an upper bound on the gradients of the warping. As we will review more in detail in Section \ref{sub:proofs}, two of the present authors proved in \cite{deponti-mondino} that the classical inequalities \eqref{eq:cheeger0} and \eqref{eq:buser0} can be generalized as well.
More precisely, when the measure induced by the Riemannian metric $\bar{g}$ is weighted by a function $\ee^f := \ee^{(D-2)A}$, \eqref{eq:defChConst} becomes:
\begin{equation}\label{eq:h}
	h_1 := \inf_{B} \frac{\text{Vol}_f(\partial B)}{\text{Vol}_f(B)}:=
	\inf_{B} \frac{\int_{\partial B} \ee^{f}\dd\overline{\mathrm{vol}}_{n-1}}{\int_{B}\ee^{f} \dd\overline{\mathrm{vol}}_n}= \inf_{B} \frac{\int_{\partial B} \ee^{(d-1)A}\dd\mathrm{vol}_{n-1}}{\int_{B}\ee^{(d-2)A} \dd\mathrm{vol}_n}\,,
\end{equation}
generalizing \eqref{eq:h1}. Again $B$ should be such that $0<\mathrm{Vol}_f (B)\leqslant \frac12 \mathrm{Vol}_f(M_n)$.

Then, calling $K$ the constant that bounds the Bakry--\'Emery-Ricci curvature from below (which is $K = |\Lambda|+\frac{\sigma^2}{D-2}$ in  \eqref{eq:RiccifBound}), \cite{deponti-mondino} showed that
\begin{equation}\label{eq:bound}
	 \frac14 h_1^2 \leqslant m_1^2 \leqslant \max\left\{\frac{21}{10}h_1 \sqrt{-K},\,\frac{22}{5}h_1^2\right\}	\;,
\end{equation}
where $m_1^2$ is the first non-trivial spin-two mass, i.e.~the first non-trivial eigenvalue $\lambda_1$ of the Bakry--\'Emery Laplacian \eqref{eq:BELaplacian}.
Notice that the first inequality in \eqref{eq:bound} does not require a curvature bound, generalizing \eqref{eq:cheeger0}. Moreover, as we describe more precisely in Th.~\ref{th: cheeger} and Th.~\ref{th: buser}, these results also hold in the more general $\RCD$ setting introduced in Section \ref{sub:math prel}. Following the discussion in Section \ref{sec: D-branes RCD}, this implies that \eqref{eq:bound} is also valid for compactifications with brane sources.

When the weighted volume $\int_{M_n} \ee^{f} \,\dd\overline{\text{vol}}_n$ is infinite, the massless graviton is not dynamical and \eqref{eq:bound} applies to the first mass. Similarly to the Riemannian case, $h_1$ is computed by taking the minimum among all the Borel sets with finite weighted volume, and $h_1$ in the second inequality needs to be replaced by $\frac{1}{2} h_0$.
In particular, compactification manifolds with infinite weighted volume and a very small $h_1$ provide a method to realize lower-dimensional massive theories of gravity from higher dimensional theories with a massless graviton.

To better understand the physical properties of compactifications with a light spin-two field, for example to assess whether the higher spin-two modes also have small masses or if instead a gap is possible,  we need to bound the higher eigenvalues of the Bakry--\'Emery Laplacian. We turn to this now.


\subsection{Higher eigenvalues} 
\label{sub:higher}

The notions introduced above can be generalized to higher eigenvalues.
Intuitively, the number of small necks that disconnect the manifold if completely shrunk is related to the number of small eigenvalues (cf.~Fig.~\ref{fig:h-sigma}). This notion can be formalized in terms of the multi-way isoperimetric constants \eqref{def: m-Cheeger}. For Bakry--\'Emery manifolds, where $\dd\mm=\ee^{f} \dd{\rm vol}_g$, $f= (D-2)A$, they read
\begin{equation}\label{eq:multiway}
	h_k = \inf_{B_0,\dots, B_k} \max_{0\leqslant i \leqslant k} \frac{\text{Vol}_f(\partial B_i)}{\text{Vol}_f(B_i) }:= \inf_{B_0,\dots, B_k} \max_{0\leqslant i \leqslant k} \frac{\int_{\partial B_i} \ee^f\dd \overline{\text{vol}}_{n-1}}{\int_{B_i} \ee^f \dd\overline{\text{vol}}_{n}}\;;
\end{equation}
recall that the infimum is now taken among all the collections of $k+1$ disjoint subsets $B_0, \dots, B_k$ of $M_n$ with smooth boundary and of positive measure. Roughly speaking, $h_k$ is small when $M_n$ can be disconnected in $k+1$ pieces, all of them with small necks.

For $k=1$, superficially \eqref{eq:multiway} might look different from \eqref{eq:h}: we have two $B_i$ rather than the single $B$ in \eqref{eq:h}. But we also no longer have the requirement that $\mathrm{Vol}_f (B)\leqslant \frac12\mathrm{Vol}_f(M_n)$. With this remark, $h_1$ does reduce to the (weighted) Cheeger isoperimetric constant \eqref{eq:h}.

As we will see in Th.~\ref{th: RCDUpper}, there exists a universal constant  $C>0$ (not depending on the Bakry--\'Emery manifold nor the dimension) such that the $k$-th spin-two state has a mass
\begin{equation}\label{eq:funano-lower}
	m_k \geqslant C^{-1}k^{-3} h_k \;.
\end{equation}
Thus, if the $h_k$-th isoperimetric constant is not small, the mass of the $k$-th spin-two state cannot be small either.
We will apply this result in the explicit examples below. A particularly direct application is in Section \ref{sec:riemann}, where the $h_k$ will be related to small necks arising in the decomposition of Riemann surfaces. We will also see there how it agrees with known results on the spectrum of the Laplacian on Riemann surfaces.

As in the case for the first eigenvalue, the lower bound \eqref{eq:funano-lower} does not require any assumption on the curvature. However, when a curvature bound is known, we can also put upper bounds on the higher eigenvalues, as we will prove in Th.~\ref{th: main eigen}, Th.~\ref{th: Liu est} and Th.~\ref{th: higherbuser}. For example, when Ric$_f$ is bounded from below by $K< 0$, the $k$-th spin-two mass is bounded from above as
\begin{equation}
m_k^2 < k^2\,\max\left\{-\frac{14112}{25}K,\frac{29568}{25}\sqrt{-K}h_k, \frac{61952}{25}h^2_k\right\}\;.
\end{equation}

Notice that even if we have presented the results in this section in the context of Bakry--\'Emery geometry, i.e.~for Riemannian manifolds with measure weighted by a function $\ee^f$, we will prove in Section \ref{sub:proofs} they hold in the more general $\RCD$ setting. 

Finally, combining together various isoperimetric bounds, it is also possible to directly relate the higher spin-two mode to the lightest one. For example, focusing again on the general case of a negative $K$, we have
\begin{equation}\label{eq:uppermk}
m^2_k< k^2\, \max\left\{-\frac{14112}{25}K,\frac{2816}{5}m_1^2\right\}.
\end{equation}
We refer to Th.~\ref{th: main eigen} for more details, including the infinite-volume case and stronger results when $K = 0$. 
Recalling that for compactifications of general higher-dimensional theories that satisfy the REC, $K = -(|\Lambda| + \frac{\sigma^2}{D-2})$, as in \eqref{eq:RiccifBound}, equation \eqref{eq:uppermk} provides a general upper bound on the higher massive spin-two states in terms of the lightest mode.

\subsubsection{Spin-two conjectures}
\label{ssub:conj}
	As an application of the bounds on higher eigenvalues, we will now comment on the spin-two swampland conjecture of \cite{klaewer-lust-palti} and the massive-AdS-graviton conjecture of \cite{bachas-19}.
	
The conjecture of \cite{klaewer-lust-palti} concerns in general an effective theory of gravity coupled to a massive spin-two field $w_{\mu \nu}$ of mass $m$. It states that such a theory is only sensible up to energies $\Lambda_0 = m M_\mathrm{Pl}/M_w$, where $M_w$ is a certain scale associated with the interactions of $w_{\mu \nu}$; and that beyond this scale, a tower of new fields must appear.  
In particular, in a limit where the first non-trivial KK spin-two mass $m_1$ goes to zero, a whole tower of spin-two fields with vanishing masses must appear.

The main argument presented in \cite{klaewer-lust-palti} in favor of the conjecture relies on its Eq.~(5), which is motivated in flat space.\footnote{We thank Eran Palti for illuminating correspondence regarding this point.} Thus AdS vacua might seem to be beyond its scope. However, when the cosmological constant is much smaller than the spin-two masses, we might expect the logic to still apply at least approximately; thus it makes sense to test it with our methods in this regime.

Indeed (\ref{eq:uppermk}) implies\footnote{Alternatively one can combine (2.40) and (2.41) in \cite{deluca-t-leaps}. Inverting the latter with a Lambert $W$ function, one obtains a general bound relating $m_k$ and $m_1$, which in the regime $m_1^2 > \ee^{-1} \pi^2 c(n)^2 (|\Lambda| + \frac{\sigma^2}{D-2})$ becomes similar to the one we just gave.}
\begin{equation}
	m_k^2< \frac{2816}5 k^2 m_1^2 \qquad \text{if } m_1^2 > \frac{441}{440}\left(|\Lambda| + \frac{\sigma^2}{D-2}\right)\,.
\end{equation}
In other words, if we make $m_1^2$ small but still larger than $(|\Lambda| + \frac{\sigma^2}{D-2})$, and in particular larger than $|\Lambda|$, then all the other $m_k$ are forced to  be small as well. 

Encouraged by this result, we may wonder if the conjecture also holds more generally for AdS vacua, even when $m_1^2$ becomes as small as $|\Lambda|$. The results in this paper suggest a way to find counter-examples: if we can make $h_1$ arbitrarily small, while keeping finite $h_2$ and the curvature bounds,  
the higher KK masses $m_k$, $k>1$ do not go to zero.
This is because \eqref{eq:bound} will result in an arbitrarily small $m_1$, while \eqref{eq:funano-lower} puts  a lower bound on $m_2$. In other words
\begin{eqnarray}\label{eq:conj2}
	\frac{m_2}{m_1}\gg 1\qquad\qquad \text{if }  h_1 \ll 1 \text{ with } h_2, K \text{ fixed}\;.
\end{eqnarray}
We will see explicit realizations of this mechanism in the examples in Section \ref{sec:bl} and \ref{sec:riemann}.

However, it is important to stress that the tower of massless spin-two fields predicted by the conjectures need not be the KK tower.\footnote{A different kind of subtlety is considered in \cite{bachas-19}, which suggests that the conjecture should only apply when the limit $m_1\to0$ is achieved continuously. This gives a different reason to exclude the examples of Section \ref{sec:bl}.
} Thus the conjecture might still be respected, even if $m_2/m_1\to\infty$. A natural guess is that the predicted states might be provided by branes that wrap the small neck. Indeed such states are very light in the relevant limit; if they are not BPS, they belong to long supermultiplets which also contain spin-two fields. It would be interesting to check this in detail.

\subsubsection{Continuous part of the spectrum?} 
\label{ssub:continuous}

In the presence of some D-brane sources, the Cheeger constants can actually be zero. To see this, consider a tubular neighborhood $B= \{r\leqslant R\}$, in the local coordinates of \eqref{eq:Dbrane-10}. Then 
\begin{equation}\label{eq:h1-Dp}
	\frac{\int_{\partial B} \sqrt{\bar g_{\partial B}}\,\ee^{(D-2)A}\,\dd^{n-1}x}{\int_{B} \sqrt{\bar g}\,\ee^{(D-2)A}\,\dd^{n}x}= \frac{R^{8-p}\sqrt{H(R)}}{\int_0^{R} r^{8-p} H(r)\,\dd r}\sim \frac{r_0^{(7-p)/2} R^{(9-p)/2}}{r_0^{(7-p)}\int_0^{R}r\, \dd r}\sim r_0^{(p-7)/2} R^{(5-p)/2}\,.
\end{equation}
For $p<5$, this can be made arbitrarily small by taking $r_0\to 0$; so $h_1=0$. In fact for this range of values it is easy to see that the higher $h_k$ also vanish: we may take several annular regions $B_i= \{2R_i \leqslant r \leqslant R_{i+1}\}$, with $R_i \ll R_{i+1}$ and all $R_{i+1}\to 0$.

This behavior is compatible with the presence of a continuous part of the spectrum. This is confirmed by a local analysis of solutions; for example for $p=4$ the local eigenfunctions can be written in terms of Bessel functions, and behave as $\psi\sim r^{-5/4}\cos(m (r^{-1/2}- \delta r))$, with $m$ arbitrary and $\delta r$ a constant. We think this continuous spectrum is an artifact of the supergravity approximation, which would disappear if higher-derivative corrections to the mass operators were taken into account. In any case, most AdS solutions with back-reacted D$p$-branes have $p>5$.

For $p=5$, \eqref{eq:h1-Dp} gives the constant $r_0^{-1}= l_s^{-1} g_s^{-1/2}$, which is very large when the supergravity approximation is relevant; so $h_1$ is realized by taking $B$ with a boundary not near the D5.  This case will indeed be relevant for the examples in Section \ref{sec:bl}. 

Finally for $p>5$ we see that \eqref{eq:h1-Dp} grows as $R\to 0$; so to take the infimum we want to make $R$ large, where the rest of the internal geometry comes into play, giving rise to a finite $h_1$. These are the cases where our proof of the $\RCD$ condition in Section \ref{sec: D-branes RCD} was most complicated, and to which we plan to return in the near future \cite{deluca-deponti-mondino-t2}.



\subsection{Holographic interpretation} 
\label{sub:hol}

An almost-split AdS$_d$ solution can sometimes be given a holographic interpretation. The large regions almost look like separate CFT$^a_{d-1}$ models, $a=1,2$, coupled very weakly by the bridge connecting them. The bridge itself is sometimes dual to a $d$-dimensional field theory QFT$_d$; see Fig.~\ref{fig:cft}. The latter can be conformal (Section \ref{sec:bl}) or not (Section \ref{sec:riemann}). 

\begin{figure}[ht]
	\centering
		\includegraphics[width=7cm]{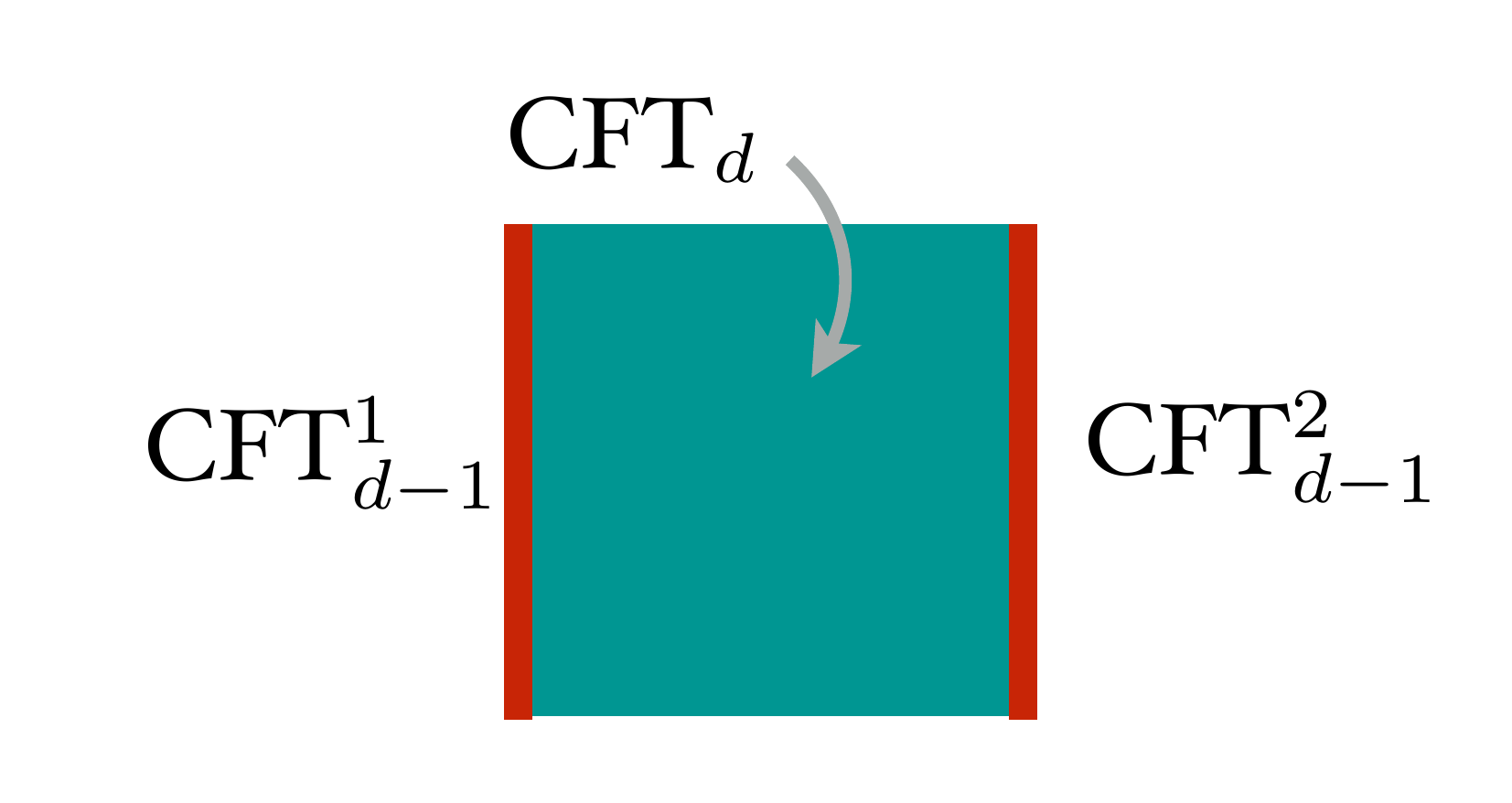}
	\caption{Sometimes an almost-split solution is dual to a CFT$_d$ (or just a QFT$_d$) on $[0,1]\times \mathbb{R}^{d-1}$, with two CFT$^a_{d-1}$, $a=1,2$ on the boundaries.}
	\label{fig:cft}
\end{figure}

There are several measures of a field theory's number of degrees of freedom. For a CFT$_d$ in $d=$ even, one often uses its Weyl anomaly coefficients. Alternatively, the free energy ${\mathcal F} = {\mathcal F}_{0} V T^d$, where $V$ and $T$ are the volume and temperature; the coefficient ${\mathcal F}_0$ is a measure of the degrees of freedom that works in any $d$, even or odd. Holographically 
\begin{equation}\label{eq:F0d}
	{\mathcal F}_0 (\mathrm{CFT}_d) \propto \int_{M_n}\ee^{(d-2)A} \dd\mathrm{vol}_n\,.
\end{equation}
This can be argued either by regarding the right-hand side as the on-shell action, or by noticing that it is proportional to the reduced Planck mass and considering a large black hole in AdS$_{d+1}$; see \cite[Sec.~4.1]{cremonesi-t} for a version of this argument.

This suggests an interpretation for the Cheeger constant \eqref{eq:h}. The denominator is the ${\mathcal F}_0$ of the CFT$_{d-1}^1$ associated to the smallest of the large regions. When the bridge is associated dual to a CFT$_d$, the numerator can also be interpreted this way, and we can write
\begin{equation}\label{eq:F0/F0}
	h \propto \frac{ {\mathcal F}_{0}(\mathrm{CFT}_d)}{{\mathcal F}_0 (\mathrm{CFT}_{d-1}^1)}\,.
\end{equation}
In the examples of Section \ref{sec:bl}, where $d=4$, $m_1^2\sim h$, and the relation of $m_1$ to (\ref{eq:F0/F0}) for $d=4$ was indeed noticed in \cite[(25)]{bachas-lavdas2}.

The case where the bridge is dual to a CFT$_d$ has several related interpretations; see \cite{bachas-lavdas} for a more thorough discussion.
\begin{itemize}
	\item It can be viewed as the CFT$_d$ living on a spacetime $I\times \mathbb{R}^{d-1}$, $I=[0,1]$, with two CFT$^a_{d-1}$, $a=1,2$ living at the two boundaries; this forms a ``CFT sandwich'', as was put in \cite{vanraamsdonk-cosmology}. The CFT$_d$ has few degrees of freedom, and operators on one of the CFT$^a_{d-1}$ can travel to the other with small probability. The individual stress-energy tensors of the two theories are almost conserved, $\partial_\mu T^{\mu \nu}_{(a)}\sim 0$; only their sum $T^{\mu \nu}:= T^{\mu \nu}_{(1)} + T^{\mu \nu}_{(2)}$ is exactly conserved. 
	At large distances, the system looks like a single CFT$_{d-1}$, whose stress-energy tensor is $T^{\mu \nu}$. 
	\item From the AdS$_{d+1}$ point of view, the two AdS$_d$ are ``ends of the world'' that meet at the boundary; this was dubbed ``wedge holography'' in \cite{akal-kusuki-takayanagi-wei}. (This was the starting point for a proposal for an accelerating cosmology in \cite{vanraamsdonk-cosmology}, and for a four-dimensional realization of the island black hole phenomenon \cite{uhlemann-island}).
	\item This class is also similar to the older double-trace deformations \cite{aharony-berkooz-silverstein,witten-multitrace,berkooz-sever-shomer}; see also \cite{bachas-19} for a comparison. In those cases, the gravity dual consists of two copies of an AdS$_d\times M_n$ solution that join at the boundary; so the AdS$_{d+1}$ region is absent.
	\item One can also view the bridge as a peculiar traversable wormhole which is accessible from everywhere in spacetime \cite{bachas-lavdas}.
\end{itemize}

In a similar manner, the non-compact space of Fig.~\ref{fig:split-nc} can also sometimes be interpreted as a CFT$_d$ on $\mathbb{R}^{d-1}\times [0, \infty)$ and a CFT$_{d-1}$ on its single boundary. This time there is no $T^{\mu \nu}$ that is exactly conserved; this is dual to the absence of a massless graviton, pointed out in the Introduction. This type of model is similar to those of Karch and Randall \cite{karch-randall,karch-randall2}.


\subsection{Proofs of eigenvalues bounds}
\label{sub:proofs}

In this section we give rigorous proofs of various results relating eigenvalues of the Laplacian and Cheeger constants. For simplicity, we will always assume to work with an $\RCD(K,\infty)$ space having eigenvalues below the infimum of the essential spectrum, and we recall the notation introduced in Section \ref{sec:background}: 
$$\Sigma:=\inf \sigma_{ess}(\Delta).$$ Nevertheless we remark that some inequalities, in particular those concerning lower bounds of the eigenvalues in terms of Cheeger constants, do not depend on the Ricci curvature lower bound and can be indeed proven under weaker assumptions on the space (see for instance \cite{Milm,deponti-mondino,DMS}).
 
First of all, we recall a partial outcome of \cite{deponti-mondino}, to which we refer the interested reader for all the details and the proofs: the first result corresponds to the classical Cheeger inequality \cite[Theorem 3.6]{deponti-mondino}; the second one to the Buser's inequality \cite[Corollary 1.2]{deponti-mondino}. 

\begin{theorem}\label{th: cheeger}
Let $(X,\mathsf{d},\mm)$ be an $\mathsf{RCD}(K,\infty)$ space, $K\in \R$. Let us suppose $\lambda_1<\Sigma$ (if $\mm(X)<\infty$) or $\lambda_0<\Sigma$ (if $\mm(X)=\infty$). We have: 
\begin{itemize}
\item If $\mm(X)<\infty$, then $h_1(X)\leqslant 2\sqrt{\lambda_1}$.
\item If $\mm(X)=\infty$, then $h_0(X)\leqslant 2\sqrt{\lambda_0}$.
\end{itemize}
\end{theorem}

\begin{theorem}\label{th: buser}
Let $(X,\mathsf{d},\mm)$ be an $\mathsf{RCD}(K,\infty)$ space, $K\leqslant 0$. Let us suppose $\lambda_1<\Sigma$ (if $\mm(X)<\infty$) or $\lambda_0<\Sigma$ (if $\mm(X)=\infty$).
\begin{itemize}
\item Case $K=0$, $\mm(X)< \infty$. It holds
\begin{equation}\label{eq:K0m1}
\lambda_{1}< \pi  h_1(X)^2.
\end{equation}
In case $\mm(X)=\infty$, the estimate \eqref{eq:K0m1} holds replacing $\lambda_{1}$ with $\lambda_{0}$ and $h_1(X)$ with $h_0(X)/2$.

\item Case $K<0$, $\mm(X)< \infty$. It holds
\begin{equation}\label{eq:K<0m1}
\lambda_1 <\max\left\{\frac{21}{10} \sqrt{-K}h_1(X),\frac{22}{5}h_1(X)^2 \right\}. 
\end{equation}
In case $\mm(X)=\infty$, the estimate \eqref{eq:K<0m1} holds replacing $\lambda_{1}$ with $\lambda_{0}$ and $h_1(X)$ with $h_0(X)/2$.
\end{itemize}
\end{theorem}
\begin{remark}
The paper \cite{deponti-mondino} contains slightly better estimates than those stated in Theorem \ref{th: buser}, with an implicit version of the Buser's inequality that makes it sharp for $K>0$. In the present article we are only interested in spaces with Ricci bounded below by a non-positive constant, and we have decided to give the statement of the Buser's inequality in the form above to make the bounds more apparent.   
\end{remark}

We are now interested in similar bounds involving higher order eigenvalues and Cheeger constants. We start with the following result that has been originally proven by Liu in the context of compact Riemannian manifolds with non-negative Ricci curvature \cite{Liu}:
\begin{theorem}\label{th: main eigen}
Let $(X,\di,\mm)$ be a $\RCD(K,\infty)$ space, $K\leqslant 0$. Let $k\in \N^{+}$ and let us suppose that $\lambda_k<\Sigma$. 
\begin{itemize}
\item Case $K=0$. If $\mm(X)<\infty$, then 
\begin{equation}\label{eq: eigen1 mainbound K=0fin}
\lambda_k<128\pi k^2\lambda_1.
\end{equation}
In case $\mm(X)=\infty$, we have
\begin{equation}\label{eq: eigen0 mainbound K=0inf}
\lambda_k<32\pi k^2\lambda_0.
\end{equation}
\item Case $K<0$. If $\mm(X)<\infty$, then 
\begin{equation}\label{eq: eigen mainbound K<0fin}
\lambda_k< \max\left\{-\frac{14112}{25}Kk^2,\frac{2816}{5}k^2\lambda_1\right\}.
\end{equation}
In case $\mm(X)=\infty$, we have
\begin{equation}\label{eq: eigen0 mainbound K<0inf}
\lambda_k<\max\left\{-\frac{3528}{25}Kk^2,\frac{704}{5}k^2\lambda_0\right\}.
\end{equation}
\end{itemize}
\end{theorem}

\begin{remark}
As already noted by Liu in \cite{Liu} (see Examples $1.3$ and $1.4$ therein), the quadratic dependence in $k$ of the bounds given above is optimal.
\end{remark}
\begin{remark}
We remark that the constant appearing in \eqref{eq: eigen1 mainbound K=0fin} is better than the one obtained in \cite{Liu}. 
\end{remark}

Actually, Theorem \ref{th: main eigen} will be a direct consequence of Buser's inequality and the following stronger result, which has been proved in \cite[Theorem 1.6]{Liu} for compact Riemannian manifolds and then extended to complete Riemannian manifolds in \cite[Theorem 1.4]{KLP}.
\begin{theorem}\label{th: Liu est}
Let $(X,\di,\mm)$ be a $\RCD(K,\infty)$ space. Let $k\in \N^{+}$ and let us suppose that $\lambda_k<\Sigma$. If $\mm(X)< \infty$, then
\begin{equation}\label{eq: Liu est}
h_1(X)^2\lambda_k\leqslant 128k^2\lambda_1^2\, .
\end{equation}
In case $\mm(X)=\infty$, the estimate \eqref{eq: Liu est} holds replacing $\lambda_{1}$ with $\lambda_{0}$ and $h_1(X)$ with $h_0(X)$.
\end{theorem}
To prove Theorem \ref{th: Liu est} we adapt the strategy originally proposed by Liu to the general setting of $\RCD$ spaces. We remark that we are closely following the arguments of \cite{Liu,KLP}, to which we refer for more details. Here we only give a sketch of the proof emphasizing the modifications needed in the possibly non-smooth framework.

\begin{proof}[Proof of Theorem~\ref{th: Liu est}]
Let $f\in \Lip_{bs}(X)$ be non-negative, $f\not\equiv 0$. We denote by 
$$\phi(f):=\inf_{t\geqslant 0} \frac{\Per(\{x: f(x)>t\})}{\mm(\{x: f(x)>t\})}.$$
The result is a consequence of the following inequality
\begin{equation}\label{eq: fundamental estimate}
\phi(f)\leqslant 8\sqrt{2}\frac{k}{\sqrt{\lambda_k}}\frac{\|\nabla f\|_{L^2}^2}{\|f\|^2_{L^2}},  \quad \text{for all } k\in \N.
\end{equation}
Indeed, Theorem \ref{th: Liu est} follows by the definition of the Cheeger constant and by applying the estimate \eqref{eq: fundamental estimate}:
\begin{itemize}
\item to a sequence of functions $f_n\in \Lip_{bs}(X)$ converging in $L^2(X,\mm)$ to an eigenfunction $f$ of eigenvalue $\lambda_0$ with $|\nabla f_n|\rightarrow |Df|$ in $L^2(X,\mm)$, if $\mm(X)=\infty$;
\item to two sequences of functions $f_n,h_n\in \Lip_{bs}(X)$ converging in $L^2(X,\mm)$ respectively to the positive and negative parts $f^{+},f^{-}$ of an eigenfunction $f$ of eigenvalue $\lambda_1$ with $|\nabla f_n|\rightarrow |Df^{+}|$ and $|\nabla h_n|\rightarrow |Df^{-}|$ in $L^2(X,\mm)$, if $\mm(X)<\infty$.
\end{itemize}
We recall here that the existence of the sequences with the above stated properties is a consequence of the density of $\Lip_{bs}(X)$ in $W^{1,2}(X,\di,\mm)$.

It is thus sufficient to prove \eqref{eq: fundamental estimate} for any non-negative function $f\in \Lip_{bs}(X)$, $f\not\equiv 0$. 

Given a finite set $A$ of real numbers, we define the function $\psi_A:\R\rightarrow \R$ as
$$\psi_A(s):=\arg \min_{t\in A}|s-t|.$$
We also set 
$$\eta_A:\R\rightarrow \R \qquad \eta_A(s):=|s-\psi_A(s)|$$
and
$$\eta_{A,f}:X\rightarrow \R \qquad \eta_{A,f}(x):=\eta_A\circ f(x)=|f(x)-\psi_A(f(x))|.$$

We now fix $k\in \N$ and we define by induction a sequence $\{t_j\}$ such that $t_0=0$ and, given $t_0<t_1<...<t_{j-1}$, the number $t_j$ is defined as the smallest $t>t_{j-1}$ such that
\begin{equation}\label{eq: induction def}
\|\eta_{\{t_{j-1},t\},f}\chi_{f^{-1}((t_{j-1},t])}\|_{L^2}^2=\frac{1}{k\lambda_k}\|\nabla f\|_{L^2}^2=:C_0
\end{equation}
if such a $t$ exists, and $t_j=\|f\|_{L^{\infty}}$ otherwise.
Denoting by 
$$f_j:=\eta_{\{t_{j-1},t_j\},f}\chi_{f^{-1}((t_{j-1},t_j])}, \qquad j\geqslant 1,$$
we notice that $\{f_j\}$ are a family of non-negative, Lipschitz functions in $L^2(X,\mm)$ (trivial if $t_{j-1}=\|f\|_{L^{\infty}}$) such that $\{f_{j_1}>0\}\cap\{f_{j_2}>0\}=\emptyset$ whenever $j_1\neq j_2$. Moreover, for every $x,y\in X$ and $j\geqslant 1$ it holds  $|f_j(x)-f_j(y)|\leqslant |f(x)-f(y)|$ so that 
$$\sum_{j=1}^{\infty} |\nabla f_j|^2\leqslant |\nabla f|^2,$$
and in particular $|\nabla f_j|\in L^2(X,\mm)$. 

We also notice that $t_{2k}=\|f\|_{L^{\infty}}$. Indeed, if this is not the case, the fact that $t_{2k}<\|f\|_{L^{\infty}}$ and the inequality
$$\sum_{j=1}^{2k}\mathcal{R}(f_j)\leqslant \frac{1}{C_0}\|\nabla f\|_{L^2}^2=k\lambda_k$$
imply the existence of at least $k+1$ non-constant functions $\{f_j\}$ such that $\mathcal{R}(f_j)\leqslant \lambda_k$, which contradicts the min-max characterization of the eigenvalues.

We can thus consider the set $A:=\{0=t_0<t_1\leqslant ...\leqslant t_{2k}=\|f\|_{L^{\infty}}\}.$ By the above considerations we know that 
\begin{equation}\label{eq: 0final}
\|\eta_{A,f}\|_{L^2}^2\leqslant \frac{2}{\lambda_k}\|\nabla f\|_{L^2}^2.
\end{equation}

We now introduce the function $g:X\rightarrow \R$ defined as
$$g(x):=\int_0^{f(x)}\eta_A(t)\,\dd t.$$
We notice that $g\in \Lip_{bs}(X)$, it has the same level sets of $f$ since $g(x)>g(y)$ if and only if $f(x)>f(y)$ and, by applying the co-area inequality for Lipschitz functions on metric measure space (see for instance \cite[Proposition 4.1]{deponti-mondino}), it holds 
\begin{equation}\label{eq: 1final}
\phi(f)=\phi(g)\leqslant \frac{\|\nabla g\|_{L^1}}{\|g\|_{L^1}}.
\end{equation}
By the chain rule for Lipschitz functions, the fundamental theorem of calculus and the Cauchy--Schwarz inequality we can infer
\begin{equation}\label{eq: 2final}
\|\nabla g\|_{L^1}\leqslant \|\nabla f\|_{L^2}\|\eta_{A,f}\|_{L^2}.
\end{equation}

Finally, we also have at our disposal the pointwise estimate
\begin{equation}\label{eq: 3final}
g\geqslant \frac{1}{8k}f^2
\end{equation} 
which can be derived by elementary considerations using the definition of $g$ and $\eta_A$. 

We are now ready to conclude putting together \eqref{eq: 0final}, \eqref{eq: 1final},\eqref{eq: 2final} and \eqref{eq: 3final} obtaining
$$\phi(f)\leqslant \frac{\|\nabla g\|_{L^1}}{\|g\|_{L^1}}\leqslant 8k\frac{\|\nabla f\|_{L^2}\|\eta_{A,f}\|_{L^2}}{\|f\|_{L^2}^2}\leqslant 8\sqrt{2}\frac{k}{\sqrt{\lambda_k}}\frac{\|\nabla f\|_{L^2}^2}{\|f\|^2_{L^2}}.$$
\end{proof} 

With these results at our disposal, it is easy to derive a higher order Buser inequality.  
\begin{theorem}\label{th: higherbuser}
Let $(X,\di,\mm)$ be a $\RCD(K,\infty)$ space. Let $k\in \N^{+}$ and let us suppose that $\lambda_k<\Sigma$. Then 
\begin{itemize}
\item Case $K=0$. If $\mm(X)<\infty$, then
\begin{equation}\label{eq: higherbuser K=0fin}
\lambda_k< 128\pi^2 k^2h_k(X)^2.
\end{equation}
In case $\mm(X)=\infty$, we have
\begin{equation}\label{eq: higherbuser K=0infin}
\lambda_k< 8\pi^2 k^2h_k(X)^2.
\end{equation}

\item Case $K<0$. If $\mm(X)<\infty$, then
\begin{equation}\label{eq: higherbuser K<0fin}
\lambda_k< \max\left\{-\frac{14112}{25}Kk^2,\frac{29568}{25}k^2\sqrt{-K}h_k(X), \frac{61952}{25}k^2h^2_k(X)\right\}.
\end{equation}
In case $\mm(X)=\infty$, we have
\begin{equation}\label{eq: higherbuser K<0infin}
\lambda_k< \max\left\{-\frac{3528}{25}Kk^2,\frac{3696}{25}k^2\sqrt{-K}h_k(X), \frac{3872}{25}k^2h^2_k(X)\right\}.
\end{equation}
\end{itemize}
\end{theorem}
\begin{proof}
The proof is a direct consequence of Theorem \ref{th: main eigen}, Theorem \ref{th: buser} and the fact that $h_1(X)\leqslant h_k(X)$.
\end{proof}

It is interesting that a higher order Cheeger inequality is also valid. This has been firstly noticed in the setting of finite graphs by Lee, Gharan, and Trevisan in \cite{lee-gharan-trevisan}, and then extended to compact Riemannian manifolds by Miclo \cite[Theorem 7]{Miclo} (see also \cite{funano}). We are going to prove here that the result remains valid for $\RCD$ spaces, even of infinite measure, remarking that some additional work is needed here with respect to the work of Miclo, due to the lack of smoothness. The interested reader can compare our proof with the arguments contained in \cite[page 326]{Miclo}.

\begin{theorem}\label{th: RCDUpper}
There exists an absolute constant $C>0$ such that for any metric measure space $(X,\di,\mm)$ satisfying the $\RCD(K,\infty)$ condition, and for any $k\in \N^{+}$ such that $\lambda_k<\Sigma$, it holds
\begin{equation}
h_k(X)^2\leqslant Ck^6\lambda_k.
\end{equation}
\end{theorem}
\begin{proof}
Recall  that on a $\RCD(K,\infty)$ space the Laplacian corresponds to the generator of the Markovian heat semigroup.  We divide the proof in two cases:
\begin{itemize}
\item Case $\mm(X)<\infty$:

Since the measure is finite, we are in position to apply directly a result of Miclo \cite[page 325]{Miclo} and infer the existence of an absolute constant $\tilde{C}>0$ such that
\begin{equation}\label{eq: Miclo Lambda_m-lambda_m}
\frac{\tilde{C}}{k^6}\Lambda_k\leqslant \lambda_k
\end{equation}
where
$$\Lambda_k:=\min\left\{\max_{j}\,  \lambda_0(B_j)\,:\, (B_0,...,B_k) \text{ are pairwise disjoint Borel sets},\, \mm(B_j)>0\right\}, $$
and $\lambda_0(B)$ is defined as
$$\lambda_0(B):=\inf\left\{ \mathcal{R}(f)\, : \, f\in W^{1,2}(X,\di,\mm),\, f=0\,\,  \mm\textrm{-a.e.~on}\, B^c, f\not\equiv 0\right\}.$$

We fix now a Borel set $B\subset X$ with $\mm(B)\in (0,\mm(X))$, $\lambda_0(B)>0$. By definition, for any constant $P>1$ there exists a non-null function $f\in W^{1,2}(X,\di,\mm)$, $f=0$ $\mm$-a.e.~on $B^c$, such that 
$$2P\sqrt{\lambda_0(B)}\geqslant 2\left(\frac{\int_X |Df|^2\,\dd\mm}{\int_X |f|^2\,\dd\mm}\right)^{1/2}.$$
Setting $B_t:=\{x\in X : |f^2(x)|\geqslant t\}\subset B$, $t>0$, and reasoning as in the proof of \cite[Theorem 4.6]{DMS} (using the fact that $f^2$ is a $\BV$ function to which we can apply the co-area formula) we can conclude that
\begin{equation}\label{eq: first part cheeger}
 2P\sqrt{\lambda_0(B)}\geqslant \frac{\int_0^{\infty}\Per(B_t)\,\dd t}{\int_0^{\infty}\mm(B_t)\,\dd t}.
\end{equation}

We denote by $\phi_2(f):=\inf_{t>0} \frac{\Per(B_t)}{\mm(B_t)}$ and claim that
\begin{equation}\label{eq:phi2f>0}
\phi_2(f)>0\,.
\end{equation}
First, observe that
\begin{equation}\label{eq:mBt<12}
 \inf \bigg\{ \frac{\Per(B_t)}{\mm(B_t)} \,:\,  \mm(B_t)\in (0, \mm(X)/2] \bigg\}\geqslant h_{1}(X)>0 \, ,
 \end{equation}
where the first inequality follows by the very definition of Cheeger constant \eqref{eq:defChConst}, while the second is a consequence of the Buser's inequality (see Theorem \ref{th: buser}) and the fact that $\lambda_1(X)>0$.
\\We now prove a lower bound on  $\Per(B_t)/\mm(B_t)$, in case $\mm(B_t)> \mm(X)/2$. Since $\mm(B^{c}_{t})\geq \mm(B^{c})>0$, we  have that
\begin{align}
\frac{\Per(B_t)}{\mm(B_t)} &\geq \frac{\Per(B_t^{c})}{\mm(B) } =  \frac{\mm(B_t^{c})}{\mm(B) }  \;  \frac{\Per(B_t^{c})} {\mm(B_t^{c}) } \nonumber \\
& \geq  \frac{\mm(B^{c})}{\mm(B)} \, h_{1}(X), \quad \text{for all $t>0$ s.t. $\mm(B_t)>\frac{\mm(X)}{2}$.} \label{eq:mBt>12}
\end{align} 
The claim \eqref{eq:phi2f>0} follows by \eqref{eq:mBt<12} and \eqref{eq:mBt>12}.

 For any $Q>1$, we can thus find a $\bar{t}\in (0,\infty)$ such that $\mm(B_{\bar{t}})>0$ and
\begin{equation}\label{eq: second part cheeger}
Q\phi_2(f)\geqslant \frac{\Per(B_{\bar{t}})}{\mm(B_{\bar{t}})}.
\end{equation}
Since it is trivial that 
$$\displaystyle \frac{\int_0^{\infty}\Per(B_t)\,\dd t}{\int_0^{\infty}\mm(B_t)\,\dd t}\geqslant \phi_2(f),$$
by \eqref{eq: first part cheeger} and \eqref{eq: second part cheeger} it follows that for any Borel set $B\subset X$, $\mm(B)\in (0,\mm(X))$, and for any $P,Q>1$, there exists a set $B_{\bar{t}}\subset B$, $\mm(B_{\bar{t}})\in (0,\mm(X))$, such that
\begin{equation}\label{eq: third part cheeger}
2PQ\sqrt{\lambda_0(B)}\geqslant \frac{\Per(B_{\bar{t}})}{\mm(B_{\bar{t}})}
\end{equation}
Since $B\subset X$ and $P,Q>1$ are arbitrary, using \eqref{eq: third part cheeger} we can infer that
$h^2_k(X)\leqslant 4\Lambda_k$ by the very definition \eqref{def: m-Cheeger} of $h_k(X)$. This last fact together with \eqref{eq: Miclo Lambda_m-lambda_m} leads to the desired conclusion, setting $C:=4/\tilde{C}.$
\item Case $\mm(X)=+\infty$:

We reason by approximation to prove \eqref{eq: Miclo Lambda_m-lambda_m} also for spaces with infinite measure. Then the proof proceeds as above, and actually it is even easier. 

The inequality \eqref{eq: Miclo Lambda_m-lambda_m} can be proven as follows: given a non-negligible Borel set $E\subset X$ with finite measure, we introduce the notation $H_{t,E}$ for the heat semigroup restricted to $E$:
\[
H_{t,E}: L^{2}(\mm_{E})\to L^{2}(\mm_{E}), \quad H_{t,E}(f):= \chi_{E} H_{t}(\chi_{E} \, f),
\]
where $\mm_{E}:= \mm(E)^{-1} \, \mm \llcorner_{E}$  is the conditional expectation of $\mm$ with respect to $E$.
 $H_{t,E}$ is a continuous self-adjoint semigroup in $L^2(\mm_E)$ with generator denoted by $\Delta_E$ (see \cite[page 325]{Miclo} for all the details). We also denote by 
\begin{align*}
&\Lambda_k(E):= \\
&\min\left\{\max_{j}\,  \lambda_0(B_j)\,:\, (B_0,...,B_k) \text{ are pairwise disjoint Borel sets},\, \mm(B_j)\in(0,\infty)\right\},
\end{align*}
and by 
$$\lambda_k(E):= \min_{V_{k+1}(E)}\, \max_{f\in V_{k+1}(E),\\ f\not\equiv 0} \ \mathcal{R}(f)\, ,$$
where $V_{k+1}(E)$ denotes a $(k+1)$-dimensional subspace of the vector space of functions $f\in W^{1,2}(X,\di,\mm)$ such that $f\equiv 0$ on $E^{c}$.

By the above mentioned result of Miclo \cite[page 325]{Miclo}, we know that there exists an absolute constant $\tilde{C}$ such that for any set $E$ of finite measure we have
$$\frac{\tilde{C}}{k^6}\Lambda_k\le \frac{\tilde{C}}{k^6}\Lambda_k(E)\le \lambda_k(E),$$
where the first inequality trivially follows from the definition of $\Lambda_k$ and $\Lambda_k(E)$.
The result is thus proven if for any $k\in \N^{+}$ and for any $\varepsilon>0$ we find a Borel set $E$, with $\mm(E)\in (0,\infty)$, such that $\lambda_k(E)\le \lambda_k+\varepsilon$. 
To prove the last statement, we fix an eigenvalue $\lambda_k$ below the infimum of the essential spectrum of $\Delta$. By the min-max characterization, we know that there exist a $(k+1)$-dimensional subspace $V_{k+1}$ of $W^{1,2}(X,\di,\mm)$ and an $L^2(X,\mm)$ orthonormal basis $(f_0,\ldots, f_k)$ of $V_{k+1}$ such that
$$\lambda_k\ge \int_X |Df_i|^2\,\dd \mm \ \ \ \forall i=0,\ldots,k \qquad \mathcal{E}(f_i,f_j)=0\ \ \ \forall i\neq j=0,\ldots,k$$
where $\mathcal{E}$ is the symmetric bilinear form associated to the Cheeger energy.

By the density of $\Lip_{bs}(X)$ in $W^{1,2}(X,\di,\mm)$, for every $\delta>0$ we can find $f_{i,\delta}\in \Lip_{bs}(X)$, $i=0,\ldots,k$, such that $\int_X f_{i,\delta}^2\,\dd\mm=1$ and 
\begin{equation}\label{eq:relat. eigenfunctions}
\left|\int_Xf_{i,\delta}f_{j,\delta}\dd\mm\right|+|\mathcal{E}(f_{i,\delta},f_{j,\delta})|\le\delta \ \ \forall i\neq j, \qquad \int_X |Df_{i,\delta}|^2\,\dd\mm\le \lambda_k+\delta \ \ \forall i.
\end{equation}
We thus define $E:=\bigcup_i \supp(f_{i,\delta})$. Called $V_{E,k}:=\textrm{span}(f_{0,\delta}, \ldots, f_{k,\delta})$, using \eqref{eq:relat. eigenfunctions}, it is easily checked that $\dim(V_{E,k})=k+1$ and for every $f\in V_{E,k}$ with $\|f\|_{L^{2}(\mm_{E})}=1$ it holds that $\mathcal{E}(f,f)\leq\lambda_k+\varepsilon(\delta | k)$ with $\varepsilon(\delta | k)\rightarrow 0 $ as $\delta\rightarrow 0$ for every $k$. It follows that  $\lambda_{k}(E)\leq \lambda_{k}+ \varepsilon(\delta | k)$, as desired.
\end{itemize}
\end{proof}


\section{AdS$_4$ vacua with ${\mathcal N}=4$ supersymmetry} 
\label{sec:bl}

Our most important class of examples will be the one that motivated this paper, where the presence of a light graviton was found explicitly in \cite{bachas-lavdas,bachas-lavdas2}.

\subsection{The general ${\mathcal N}=4$ class} 
\label{sub:rev}

The relevant solutions are written as a fibration of $\mathbb{S}^2\times \mathbb{S}^2$ over a strip $\mathbb{R}\times [0,\pi/2]$, with coordinates $x$, $y$. There is ${\mathcal N}=4$ supersymmetry, whose R-symmetry $\mathrm{so}(4)\cong \mathrm{su}(2) \oplus  \mathrm{su}(2)$ rotates the two $\mathbb{S}^2$'s. All fields can be written in terms of two real harmonic functions $H_1$, $H_2$ on the strip, such that $H_1= \partial_y H_2=0$ at $y=0$, and $H_2 = \partial_y H_1=0$ at $y=\pi/2$. All we really need for our purposes is the metric, which we give in Einstein frame, with $z:= x + \ii y$:
\begin{align}\label{eq:ads4s2s2}
	&\dd s^2_{10}=\frac{2(N_1 N_2)^{1/4}}{\sqrt{W}}
	\Big(\dd s^2_{\mathrm{AdS}_4} + \frac{2W}{H_1 H_2} \dd z \dd \bar z
	+ \frac{H_1^2 W}{N_1}\dd s^2_{\mathbb{S}^2_1} + \frac{H_2^2 W }{N_2}\dd s^2_{\mathbb{S}^2_2}\Big)\,,\\
	\nonumber &W= \partial_z \partial_{\bar z} (H_1 H_2) \, ,\qquad N_1 = 2 H_1 H_2 |\partial_z H_1|^2 -H_1^2 W \, ,\qquad
	 N_2 = 2 H_1 H_2 |\partial_z H_2|^2 -H_2^2 W\,.
\end{align}
The AdS$_4$ metric has radius one. The barred metric $\dd \bar s^2_6$ is the internal part of the expression inside the parenthesis. We also mention that the string coupling is given by $\ee^\phi=(N_2/N_1)^{1/2}$.\footnote{All other fields are also known explicitly; see \cite{dhoker-estes-gutperle,dhoker-estes-gutperle2,assel-bachas-estes-gomis}. Notice however that these references use a different normalization of the dilaton, which differs by a factor of two.} 

The metrics of the two $\mathbb{S}^2$s are round. The $\mathbb{S}^2_1$ and $\mathbb{S}^2_2$ shrink at $y=0$ and $y=\pi/2$ respectively; so each locus $\{x=x_0\}$ is topologically an $\mathbb{S}^5$. But $x$ ranges over $\mathbb{R}$, and this appears to make $M_6$ non-compact. Indeed this class was originally found in \cite{dhoker-estes-gutperle,dhoker-estes-gutperle2} as the gravity dual to interfaces in ${\mathcal N}=4$ super-Yang--Mills (YM). For a generic choice of $H_a$, the two limits $x\to \pm \infty$ would reconstruct two AdS$_5$ limits; the central region would represent the degrees of freedom of the interface. However, with the particular choice \cite{assel-bachas-estes-gomis}
\begin{equation}\label{eq:h12-assel}
	H_1= -2\mathrm{Re}\sum_a \gamma_a \log \tanh\frac{\pi\ii + 2\delta_a - 2z}4
	\, ,\qquad
	H_2 = -2\mathrm{Re}\sum_a \hat \gamma_a \log \tanh\frac{z- \hat \delta_a}2 \, ,
\end{equation}
the $\mathbb{S}^5$ shrinks at both $x\to \pm \infty$, which are moreover at finite distance. Hence $M_6\cong \mathbb{S}^6$. So with this choice the AdS$_5$ regions have been pinched off, obtaining a compact internal space; this is the gravity dual of focusing on the CFT$_3$ that lives on the interface, decoupling the two CFT$_4$ sides. At the special points $z=\hat\delta_a$, $z=\delta_a +\frac\pi2\ii$, there are singularities, which can be identified with the behavior of NS5-branes and D5-branes respectively. These locations are discretized by imposing that the numbers of these branes (and the $F_5$ flux  quantum) are integer.

The holographic duality of these solutions with field theory was described in detail in \cite{assel-bachas-estes-gomis}. The solutions are viewed as the near-horizon limit of a system of D3-, D5-  and NS5-branes. If the D3-branes are made to never end on any D5-brane, one reads off the CFT$_3$ as a chain of $\mathrm{SU}(N_i)$ gauge groups, with hypermultiplets in fundamental and bifundamental representations \cite{hanany-witten}. One can also view these theories as arising from the $d=4$, ${\mathcal N}=4$ super-YM living on the D3s on a spacetime $I \times \mathbb{R}^3$, with $I$ an interval and different boundary conditions on the two sides. This point of view can be exploited also by placing all NS5-branes on one side and all D5-branes on the other \cite{gaiotto-witten-1}.


\subsection{Almost-split symmetric case} 
\label{sub:ass}

We now specialize to almost-split solutions as in \cite{bachas-lavdas,bachas-lavdas2}. These are obtained by dividing the NS5 and D5 locations $\delta_a$, $\hat \delta_a$ in \eqref{eq:h12-assel} in two groups, separated in the $x$ direction by a long region $\Delta x := \xi \gg 1$ without any branes. We will first consider the simpler case where the two sides are symmetric; in Section \ref{sub:asas} we will show the modifications required for the asymmetric case.

From the dual CFT point of view, the chain of gauge groups gets naturally divided in two chains, connected by a single $\mathrm{SU}(n)$, with $n$ much smaller than the other gauge groups. The two chains are then viewed as two separate CFT$_3$s which are almost decoupled, only communicating with each other through a small ``portal''. The $\mathrm{SU}(n)$ can also be thought of as arising from a $d=4$ ${\mathcal N}=4$ super-YM on  $I \times \mathbb{R}^3$, which at large distances reduces to a $d=3$ gauge theory. This realizes the field theory part of the discussion in Section \ref{sub:hol}. 

As in \cite{bachas-lavdas}, we will focus on a particularly simple almost-split configuration. There are $N-1$ D5-branes at both loci $z=(\pm (\xi-u)+\ii \pi)/2$, and one D5-brane each at $z=\pm (\xi-u)/2 -\mp \delta u+\ii \pi/2$; and $\hat N$ NS5-branes at both loci $z= \pm (\xi + u)/2 \mp \delta \hat u$. The parameters are fixed by flux quantization as 
\begin{equation}\label{eq:xi}
	\xi \sim \log \frac{4 N \hat N}{\pi n} \gg 1 \, ,\qquad
	\delta u = \frac{\pi n}{\hat N s_0} \, ,\qquad \delta \hat u = \frac{\pi n}{N \hat N s_0}\,,
\end{equation}
where $s_0 = \sin (2 \arctan{\ee^{-u_0}})$. In the limit $\xi\gg 1$ the harmonic functions read 
\begin{equation}
	H_1 \sim 8 N \ee^{-\xi/2} \cosh x \sin y \, ,\qquad
	H_2 \sim 8 \hat N \ee^{-\xi/2} \cosh x \cos y
\end{equation}
in the central ``bridge'' region $-\frac \xi2 \ll x \ll \frac \xi2$. The metric in \eqref{eq:ads4s2s2} becomes
\begin{equation}\label{eq:ads5-s5}
\begin{split}
	\dd s^2_{10}&=L_5^2
	\Big(\cosh^2 x \dd s^2_{\mathrm{AdS}_4} + \dd x^2 + \dd y^2
	+ \sin^2 y\dd s^2_{\mathbb{S}^2_1} + \cos^2 y\dd s^2_{\mathbb{S}^2_2}\Big)\\
	&= L_5^2 \left(\dd s^2_{\mathrm{AdS}_5}+ \dd s^2_{\mathbb{S}^5}\right)
	\,,
\end{split}
\end{equation}
with $L_5^2= 16\ee^{-\xi/2} \sqrt{N \hat N}$.

The smallest KK spin-two mass was estimated in \cite{bachas-lavdas} for this configuration by approximating the eigenfunction $\psi$ of \eqref{eq:BELaplacian}: in the two CFT$_3$ regions it is taken to be constant, and in the AdS$_5$ region it is found explicitly by applying \eqref{eq:BELaplacian} to \eqref{eq:ads5-s5}. This corresponds to the mixing of the two massless gravitons. The final expression obtained this way is \cite{bachas-lavdas}
\begin{equation}\label{eq:m-bl}
	m_1^2 = \frac{3 \kappa_4^2 n^2}{8 \pi^2 L_4^4} \,,
\end{equation}
where $\kappa_4$ is the four-dimensional Newton's constant, and we have restored the AdS$_4$ radius $L_4$, that was set to one in \eqref{eq:ads4s2s2}.

Let us now see if this is consistent with our bound \eqref{eq:bound}. We look among the subsets $B$ of the form $B= \{ x\leqslant x_0\}$, with $x_0$ in the central region. The numerator $\mathrm{Vol}_f( \partial B)$ of \eqref{eq:h} can be calculated from \eqref{eq:ads5-s5} to be
\begin{align}
	\nonumber\int_{\partial B} \ee^{8A}\dd\overline{\mathrm{vol}}_5 
	=\int_{\partial B} \ee^{3A}\dd\mathrm{vol}_5 
	&\sim 
	2^{13} (N \hat N)^2 \ee^{-2 \xi}\cosh^3 x_0
	(\mathrm{Vol}(\mathbb{S}^2))^2\int_0^{\pi/2} \dd y \sinh^2(2y)\\
	&= 2^{11}\pi^5 n^2 \cosh^3 x_0\,.
\end{align}
The estimate
\begin{equation}
	 \int_{B}\ee^{8A}\dd\overline{ \mathrm{vol}}_6=\int_{B}\ee^{2A}\dd \mathrm{vol}_6= 2^9 \pi^2 \int \dd x \dd y H_1 H_2 \partial_z \partial_{\bar z} (H_1 H_2)\sim (N\hat N)^2 v_6
\end{equation}
was given in \cite{bachas-lavdas}, for $v_6$ an order-one constant. So we can estimate the integral in the denominator of \eqref{eq:h} as half of this, plus the contribution from the integral $\int_0^{x_0}$, which using \eqref{eq:ads5-s5} again works out to $2^{11}\pi^5 n^2 \cosh^3 x_0$. This leads to
\begin{equation}\label{eq:h-bl0}
	h_1= \mathrm{inf}_{B}\frac{\text{Vol}_f(\partial B)}{\text{Vol}_f(B)}=\mathrm{inf}_{x_0}\frac{2^{11}\pi^5 n^2 \cosh^3 x_0}{\frac12 (N \hat N)^2 v_6 - 2^{11}\pi^5 n^2 \int_0^{x_0}\dd x\cosh^2 x}\,.
\end{equation}
The minimum is obtained for $x_0=0$; so
\begin{equation}\label{eq:h-bls}
	h_1= \frac{2^{11}\pi^5 n^2 }{(N \hat N)^2 v_6}= \frac1{2\pi^2} \kappa_4^2 n^2\,,
\end{equation}
using also \cite[(5.5)]{bachas-lavdas}. 

In other words, the mass estimate \eqref{eq:m-bl} reads
\begin{equation}\label{eq:m-bl-h}
	m_1^2\sim \frac34 h_1\,.
\end{equation}
 Comparing with \eqref{eq:bound}, at the lower end the bound is satisfied because we are assuming small $h_1$, or in other words $n \ll N \hat N$; recall \eqref{eq:xi}. For the same reason at the upper end we should select $\frac{21}{10}h \sqrt{K}$; since AdS$_4$ in \eqref{eq:ads4s2s2} has radius one, $\Lambda=-3$, and \eqref{eq:m-bl-h} agrees with the bound because $\frac34 < \frac{21}{10}\sqrt3$. 

To be more precise we can estimate $\sigma$. The largest warping variation is achieved at the boundaries of the central region, where the metric is given in \eqref{eq:ads5-s5}. Here $\ee^{2A_\mathrm{E}} = n^2 \cosh^2 x$, and 
\begin{equation}\label{eq:bl-sigma}
	|\dd A_\mathrm{E}|^2 \sim \sinh^2 x\,.
\end{equation}
Since this is valid in the region $|x| \ll \frac\xi2$, we have a worst-case estimate $\sigma^2 \ll \ee^\xi \sim \frac{4 N \hat N}{\pi n}$. So overall the upper bound in \eqref{eq:bound} is proportional to $h_1^\alpha$, $\alpha<3/4$.

In conclusion, while the estimate in \cite{bachas-lavdas} is of course more precise, our general-purpose bounds do quite well here, constraining the smallest eigenvalues to be small, going to zero with $h_1\propto (n/N \hat N)^2$ with a behavior between $h_1^{3/4}$ and $h_1^2$.


\subsection{Higher eigenvalues} 
\label{sub:bl-higher}

Our methods also allow us to show that the second eigenvalue should \emph{not} go to zero in the limit $n \ll N \hat N$. 

As we anticipated in Section \ref{sub:higher}, $h_k$ is expected to be small if one can split $M_6$ in $k+1$ pieces with small necks. Since the spaces we are considering can only be split in two, we only expect the first eigenvalue to be small. 

Let us see this in more detail. We take the boundaries of the three subsets $B_i$ to be at loci with constant $x$. To minimize the contribution of boundaries, we take
\begin{equation}
	B_1 = \{ x\leqslant x_-\} \, ,\qquad B_2=\{ x_-' \leqslant x \leqslant x_+'\} \, ,\qquad
	B_3 = \{x_+ \leqslant x\}\,.  
\end{equation}
If we take the cut locus $x_-$ outside the neck region, $B_1$ will be fully in the large region on the left, and $\mathrm{Vol}_f(\partial B_1)/\mathrm{Vol}_f (B_1)$ will not be small. So let us take $x_-$ in the neck region, and $x_+$ too for the same reason. To avoid overlaps between the $B_i$, then $x_\pm'$ will also be in the neck region. 

Now we will take $x_+ = -x_-$, $x'_+= - x_-'$ for simplicity; the general case is similar. A computation similar to \eqref{eq:h-bl0} gives
\begin{equation}\label{eq:hX2}
	\frac{\mathrm{Vol}_f (\partial B_2)}{\mathrm{Vol}_f(B_2)} = 
	\frac{\cosh^3 x'_+}{\int_0^{x'_+}\dd x\cosh^2 x}
\end{equation}
for the central piece $B_2$, while for $B_1$ and $B_3$ the result is just the fraction in the right hand side of \eqref{eq:h-bl0} with $x_0 \to x_+$ (prior to taking the infimum).

Thus we know already that $\frac{\mathrm{Vol}_f (\partial B_1)}{\mathrm{Vol}_f(B_1)}= \frac{\mathrm{Vol}_f (\partial B_3)}{\mathrm{Vol}_f(B_3)}$ have an infimum for $x_+=0$, given by \eqref{eq:h-bls}, which is of order $(n/N \hat N)^2\sim \ee^{-2\xi}\ll 1$. When $x_+$ gets larger, this grows; its larger value in the neck region is at its boundary, $x_+\sim \xi/2$, where $\frac{\mathrm{Vol}_f (\partial B_1)}{\mathrm{Vol}_f(B_1)}\sim \ee^{-\xi/2}$. However, \eqref{eq:hX2} is always larger than this. It diverges when $x'_+\to 0$, it has value of order $\sim \ee^{\xi/2}$ at $x'_+= \xi/2$, and it has a minimum in between where it is of order one. Since in \eqref{eq:multiway} we are instructed to take the largest of the three $\frac{\mathrm{Vol}_f (\partial B_i)}{\mathrm{Vol}_f(B_i)}$, the final result for $h_2$ is of order one. By \eqref{eq:funano-lower}, the second eigenvalue $m_2$ is also of order one: unlike $m_1$, it cannot be made arbitrarily small in this class of solutions. 

Clearly we made a few assumptions in this computation; so the argument above cannot be regarded as a proof. However, it strongly suggests that $m_2/m_1\to \infty$ as $n/N \hat N\to 0$, giving an example of the mechanism in (\ref{eq:conj2}). (As pointed out in \cite{bachas-19}, the limit is not continuous, so it is not covered by the massive-AdS-graviton conjecture proposed there).

In Section \ref{sec:riemann} we will see a stronger argument for a different class of solutions.


\subsection{Almost-split asymmetric case} 
\label{sub:asas}

We now consider the case where the two groups of branes separated by a large $\Delta x = \xi \gg 1$ are different.

The smallest non-zero graviton mass for such cases was estimated in \cite{bachas-lavdas2} as
\begin{equation}\label{eq:m-bl2}
	m_1^2 = \frac{3 n^2}{16 \pi^2 L_4^4}(\kappa^2_{4\mathrm{L}} + \kappa^2_{4\mathrm{R}} ) J(\cosh \delta \phi)\,,
\end{equation}
where $J$ is a positive function with $J(0)=1$ and that goes monotonically down to zero at infinity: $J(\cosh \delta \phi)\sim (\log \cosh \delta \phi)^{-1} \sim (\delta \phi)^{-1}$ with $\delta \phi \to \infty$.
For $\kappa^2_{4\mathrm{L}} = \kappa^2_{4\mathrm{R}}$ this reduces to \eqref{eq:m-bl}. On the other hand, when $\kappa^2_{4\mathrm{R}}\to 0$, the Planck mass goes to infinity and the massless graviton decouples; the spin-two field with mass $m_1$ is now the lightest dynamical field.

We will now apply our general bound to this asymmetric situation. 
We focus on the simplest generalization of the situation described above \eqref{eq:xi}. We take
\begin{equation}
\begin{split}
	&N_\mathrm{R}-1 \text{ D5s at } z=\frac{\xi-u_\mathrm{R}}2 +\ii\frac\pi2 \, ,\qquad 1 \text{ D5 at } z=- \delta u+ \frac{\xi-u_\mathrm{R}}2+\ii\frac\pi2\,;\\
	&N_\mathrm{L}-1 \text{ D5s at } z=\frac{-\xi+u_\mathrm{L}}2+\ii\frac\pi2 \, ,\qquad 1 \text{ D5 at } z=\delta u+ \frac{-\xi+u_\mathrm{L}}2+\ii\frac\pi2\,;\\
	&\hat N_\mathrm{R} \text{ NS5s at } z=\frac{\xi+u_\mathrm{R}}2- \delta \hat u_\mathrm{R} \, ,\qquad \hat N_\mathrm{L} \text{ NS5s at } z=-\frac{\xi+u_\mathrm{L}}2+ \delta \hat u_\mathrm{L} \,.
\end{split}
\end{equation}
In this simple situation, 
\begin{equation}
	\kappa^2_{4\mathrm{L,R}}\sim (N_\mathrm{L,R}\hat N_\mathrm{L,R})^2\,.
\end{equation}
Modifying the logic in \cite{bachas-lavdas} appropriately we obtain
\begin{equation}
\begin{split}
	\ee^{-\xi}\sim &\frac{\pi n}{2 (\hat N_\mathrm{L} N_R + N_\mathrm{L} \hat N_R )} \, ,\qquad \delta u_\mathrm{L} \sim \frac{\pi n}{\hat N_\mathrm{L} s_{0\mathrm{L}}} \, ,\qquad  \delta u_\mathrm{R} \sim \frac{\pi n}{\hat N_\mathrm{R} s_{0\mathrm{R}}}\,;\\
	\delta \hat u_\mathrm{L} \sim& \frac{\pi n \hat N_\mathrm{R}}{ \hat N_\mathrm{L} s_{0\mathrm{L}} (\hat N_\mathrm{L} N_R + N_\mathrm{L} \hat N_R )} \, ,\qquad
	\delta \hat u_\mathrm{R} \sim \frac{\pi n \hat N_\mathrm{L}}{ \hat N_\mathrm{R} s_{0\mathrm{R}} (\hat N_\mathrm{L} N_R + N_\mathrm{L} \hat N_R )}\,,
\end{split}
\end{equation}
where $s_{0 \mathrm{L,R}} = \sin (2 \arctan{\ee^{-u_{0\mathrm{L,R}}}})$.
The harmonic functions in the central region $|x| \ll \frac \xi 2$ now read
\begin{equation}
	H_1 \sim 4\mathrm{Re} (\ii N_\mathrm{L} \ee^{-z} -\ii N_\mathrm{R} \ee^z) \ee^{-\xi/2}\, ,\qquad
	H_2 \sim 4\mathrm{Re} (\hat N_\mathrm{L} \ee^{-z} + \hat N_\mathrm{R} \ee^z) \ee^{-\xi/2}\,.
\end{equation}
In particular 
\begin{equation}
	H_1 H_2 \sim -\frac{L_5^4}{16} \left(\frac{\cosh(2 (x-\delta x))}{\cosh \delta \phi} +1 \right)\sin(2 y) \, ,\qquad W=\frac{L_5^4}{16} \sin(2y)\,,
\end{equation}
with
\begin{equation}\label{eq:dxdf}
	\ee^{4 \delta x}= \frac{N_\mathrm{L}\hat N_\mathrm{L}}{N_\mathrm{R} \hat N_\mathrm{R}} \, ,\qquad \ee^{2\delta \phi} =\frac{N_\mathrm{L} \hat N_\mathrm{R}}{\hat N_\mathrm{L} N_\mathrm{R}} \, ,\qquad L_5^4 = 2^7 (\hat N_\mathrm{L} N_R + N_\mathrm{L} \hat N_R ) \ee^{-\xi} = 2^6 \pi n \,.
\end{equation}

We again assume that the $B$ minimizing 
$\mathrm{Vol}_f(\partial B)/\mathrm{Vol}_f(B)$ are those with a boundary in the central region, $B= \{x\leqslant x_0\}$, $|x_0| \ll \xi/2$. Without loss of generality we assume $N_\mathrm{L} \hat N_\mathrm{L} < N_\mathrm{R} \hat N_\mathrm{R}$. The same logic leading to \eqref{eq:h-bls} now gives
\begin{equation}\label{eq:h-blas}
	h_1= \mathrm{inf}_{B} \frac{\text{Vol}_f(\partial B)}{\text{Vol}_f(B)}=
	\mathrm{inf}_{x_0} \frac1{\sqrt2} 
		\frac{\left(1+ \cosh^{-1}\delta \phi\cosh(2x_0)\right)^{3/2}}
		{v_\mathrm{L} \left(\frac{N_\mathrm{L}\hat N_\mathrm{L}}{n}\right)^2 +\int_{-\xi/2}^{x_0}\dd x\left(1+ \cosh^{-1}\delta \phi\cosh(2x)\right)}\,,
\end{equation}
with $v_\mathrm{L}$ a numerical factor. When $\delta \phi=0$, this reduces to \eqref{eq:h-bls}. More generally the minimum has no analytic expression, but it simplifies in appropriate limits. 

For example we may take $\delta \phi$ to be very large, where we notice that the estimate \eqref{eq:m-bl2} is brought down by an additional factor $1/ \delta \phi$. By \eqref{eq:dxdf}, this implies $N_\mathrm{L} \hat N_\mathrm{R} \gg \hat N_\mathrm{L} N_\mathrm{R}$, $\ee^\xi \sim \frac{2 N_\mathrm{L} \hat N_\mathrm{R}}{\pi n}$. When
\begin{equation}\label{eq:largephi}
	\delta \phi \gg (N_\mathrm{L}\hat N_\mathrm{L}/n)^{-2}\,,
\end{equation}
 the minimization in \eqref{eq:h-blas} is well-approximated by
\begin{equation}\label{eq:hx-largephi}
	h \sim \frac1{W_0 (\cosh \delta \phi)}
	\, ,\qquad
	\ee^{2 x_\mathrm{min}}\sim \frac{\ee^{2 \delta x + \delta \phi}}{3 \delta \phi}\,.
\end{equation}
$W_0$ is the Lambert function, defined as the positive solution to $z=W_0(z)\ee^{W_0(z)}$; for large $z$, $W_0(z)\sim \log z - \log\log z +O(z^{-1}\log\log z)$. The computation is only sensible when $x_\mathrm{min}$ is in the neck region; this implies $\delta \phi\ll \frac{N_\mathrm{L}^2 \hat N_\mathrm{R}}{n N_\mathrm{R}}$, which is not necessarily in conflict with \eqref{eq:largephi}. Finally the arguments around \eqref{eq:bl-sigma} give us $\sigma^2 \ll \ee^{\xi - 2 \delta x - \delta \phi}\sim \sqrt{ \frac{N_\mathrm{R}\hat N_\mathrm{R}}{n}}$. 

The lower bound in \eqref{eq:bound} is then of order $\frac1{\delta \phi^2}\ll \frac1{\delta \phi}\left(\frac{n}{N_\mathrm{L} \hat N_\mathrm{L}}\right)^2$, so it is compatible with \eqref{eq:m-bl2}. The upper bound in \eqref{eq:bound} is of order $\frac1{\delta \phi}\frac{N_\mathrm{R}\hat N_\mathrm{R}}{n}$, which goes as $\frac 1 {\delta \phi}$ but with a much larger coefficient than in \eqref{eq:m-bl2}; in fact this bound gets less and less useful as $N_\mathrm{R} \hat N_\mathrm{R}$ get large. 

It is also very interesting to consider 
\begin{equation}
	N_\mathrm{R} \hat N_\mathrm{R}\to \infty
\end{equation}
even without any assumptions on $\delta \phi$. In this case $M_6$ becomes non-compact, but \eqref{eq:m-bl2} remains finite and may be small. However, $\sigma$ diverges, because $|\dd A_\mathrm{E}|^2$ has the same exponential behavior \eqref{eq:bl-sigma} for large $x$, which is now valid all the way to infinity without the requirement $x\ll \frac\xi2$. So the upper bound in \eqref{eq:bound} is useless in this case. 

In any case, in Section \ref{sec:riemann} we will see infinite-volume examples where $\sigma$ and the upper bound in \eqref{eq:bound} are finite.

We did not consider higher eigenvalues in this subsection; the numerical evaluation in \cite{bachas-estes} and intuitive reasons in \cite{bachas-19} suggest that this time all $m_k\to 0$ as $\delta \phi\to 0$, in agreement with the massive-AdS-graviton conjecture proposed there.



\section{Examples with Riemann surfaces} 
\label{sec:riemann}

\subsection{Supersymmetric twisted compactifications} 
\label{sub:mn}

Various solutions with one or more light spin-two fields can be constructed from the holographic duals of compactifications of conformal theories on Riemann surfaces $\Sigma_g$ with negative curvature. Consider a SCFT$_d$ with an AdS$_{d+1}\times M_n$ gravity dual. Compactifying it on $\Sigma_g$ with a certain partial topological twist, supersymmetry is still preserved and one obtains a SCFT$_{d-2}$, whose gravity dual is now of the form $\mathrm{AdS}_{d-1}\times M_{n+2}$: the internal space $M_{n+2}$ is a fibration over $\Sigma_g$, whose fiber is a distorted version of the original $M_n$.

Several examples exist with various values of $d$ and amounts of supersymmetry; see \cite{bobev-crichigno} for a review. The original case is an AdS$_5$ ${\mathcal N}=2$ solution in M-theory, dual to a compactification of the ${\mathcal N}=(2,0)$ theory in $d=6$ on a Riemann surface \cite{maldacena-nunez}. One can lower supersymmetry by different twists \cite{maldacena-nunez,bah-beem-bobev-wecht} or by starting with an ${\mathcal N}=(1,0)$ theory \cite{afpt}. See \cite{bah-passias-weck} for AdS$_4$ examples in IIA and \cite{benini-bobev,couzens-macpherson-passias} for AdS$_3$ in IIB. Often these solutions can be obtained by uplifting gauged supergravity vacua.

One can also consider compactifications of SCFTs on hyperbolic spaces with more than two dimensions. For example AdS$_4$ solutions with internal space fibered over a $\Sigma_{d=3}$ hyperbolic space can be obtained by uplifting the solutions in \cite{pernici-sezgin}, or more generally from \cite{rota-t,10letter}. However, we will see that these cases \emph{do not} yield arbitrarily small eigenvalues. Another variation is the inclusion of defects, starting from \cite{gaiotto-maldacena}, but we are not going to consider them in what follows.

\subsubsection{Application of general bounds} 
\label{ssub:bounds-Sigma}

We can study the existence of light spin-two fields in compactifications with Riemann surfaces both from the general theorems presented in Section \ref{sec:bounds}, which constrain them in terms of the isoperimetric constant, as well as from exploiting direct theorems on the spectrum of the Laplacian on Riemann surfaces.

For concreteness, we will tailor our discussion below on the original AdS$_5$ solutions found by Maldacena and Nu\~nez in \cite{maldacena-nunez}, but the following logic applies to all compactifications on negatively-curved Riemann surfaces. 

In our language, the internal metric of \cite{maldacena-nunez} reads
\begin{equation}\label{eq:MNmetric}
	\bar{\dd  s}^2 = \frac{1}{2}  \left[ \dd  s^2_{\Sigma_g} + \dd
\theta^2 + \frac{\cos^2 \theta}{1 + \cos^2 \theta} \dd  s^2_{\mathbb{S}^2} + 2
\frac{\sin^2 \theta}{1 + \cos^2 \theta} D \phi^2 \right]\;,
\end{equation}
and $\ee^f = 2^{-\frac92}(1 + \cos^2\theta)^\frac32$ .
This six-dimensional space features a Riemann surface $\Sigma_{g}$ of negative curvature and a topological $\mathbb{S}^4$ fibered over it. More precisely, if we think of the topological $\mathbb{S}^4$ as a join of an $\mathbb{S}^2$ an $\mathbb{S}^1$ (parametrized by the coordinate $\phi$), the latter is fibered over $\Sigma_g$. Such a fibration is described by a connection $\xi$, which appears in the metric \eqref{eq:MNmetric} through the covariant derivative $D\phi := \dd\phi+ \xi$. In local Poincaré coordinates for $\Sigma_g$, $\dd  s^2_{\Sigma_g}  = y^{-2}(\dd x^2+\dd y^2)$, the connection can be written as $\xi =  \frac{\dd x}{y}$.

As a consequence of the non-trivial fibration, the metric \eqref{eq:MNmetric} is non-diagonal and the Bakry--\'Emery Laplacian \eqref{eq:BELaplacian} does not decompose into an operator on $\Sigma_g$ and one on the $\mathbb{S}^4$.
Nevertheless, a subset of its eigenmodes consists of eigenfunctions that are constant on the $\mathbb{S}^4$, for which the spin-two operator \eqref{eq:BELaplacian} reduces to a standard Laplacian on $\Sigma_g$:
\begin{equation}\label{eq:decLap0}
	\Delta_f \psi  = m^2 \psi \qquad \implies \qquad \Delta^{(\Sigma_g)} \psi = m^2 \psi \qquad\qquad\text{for }\psi = \psi(\Sigma_g)\;.
\end{equation}
We can use this observation to start analyzing the class of modes that are non-constant only along $\Sigma_g$, asking how many of them can be light. To do so, we specialize and apply our general lower bound \eqref{eq:funano-lower} to $\Delta^{(\Sigma_g)}$ as follows.
Recall that to bound the $k$-th eigenvalue from below, we need to compute the isoperimetric constant $h_k$ in \eqref{eq:multiway}, which in turn requires to split our manifold in $k+1$ pieces.
Luckily, a compact Riemann surface with genus $g$ has a natural decomposition in $2g-2$ pair of pants, whose boundaries are geodesics. Moreover, the lengths of the boundary geodesics of each pairs of pants can be any triple of positive real numbers \cite[Th.~3.1.7]{buser-book}. (They are even part of a set of coordinates on the moduli space, the so-called \emph{Fenchel--Nielsen} coordinates). In particular, they can all be arbitrarily small. 
These $2g-2$ pieces can then be taken to be the $B_i$ in \eqref{eq:multiway} and, if the boundary geodesics are all small, then all the $h_k$ for $k\leqslant 2g-3$ are small as well. By \eqref{eq:funano-lower} we see that the presence of $2g-3$ small eigenvalues is allowed. Notice that we cannot obtain more than $2g-3$ small eigenvalues because this would require us to cut at least one of the pair of pants we already have. Even if we cut it near one of its ends, where the circle is small and contributes with a small perimeter in the denominator of \eqref{eq:multiway}, the volume of this new piece is also very small. Since the volume enters in the denominator, this cut would not result in a small $h_{k+1}$. This is illustrated in Fig.~\ref{fig:h-sigma} for $g=2$, where there are $2g-2=2$ pairs of pants.

\begin{figure}[ht]
\centering	
	\subfigure[\label{fig:h1-sigma}]{\includegraphics[width=4.5cm]{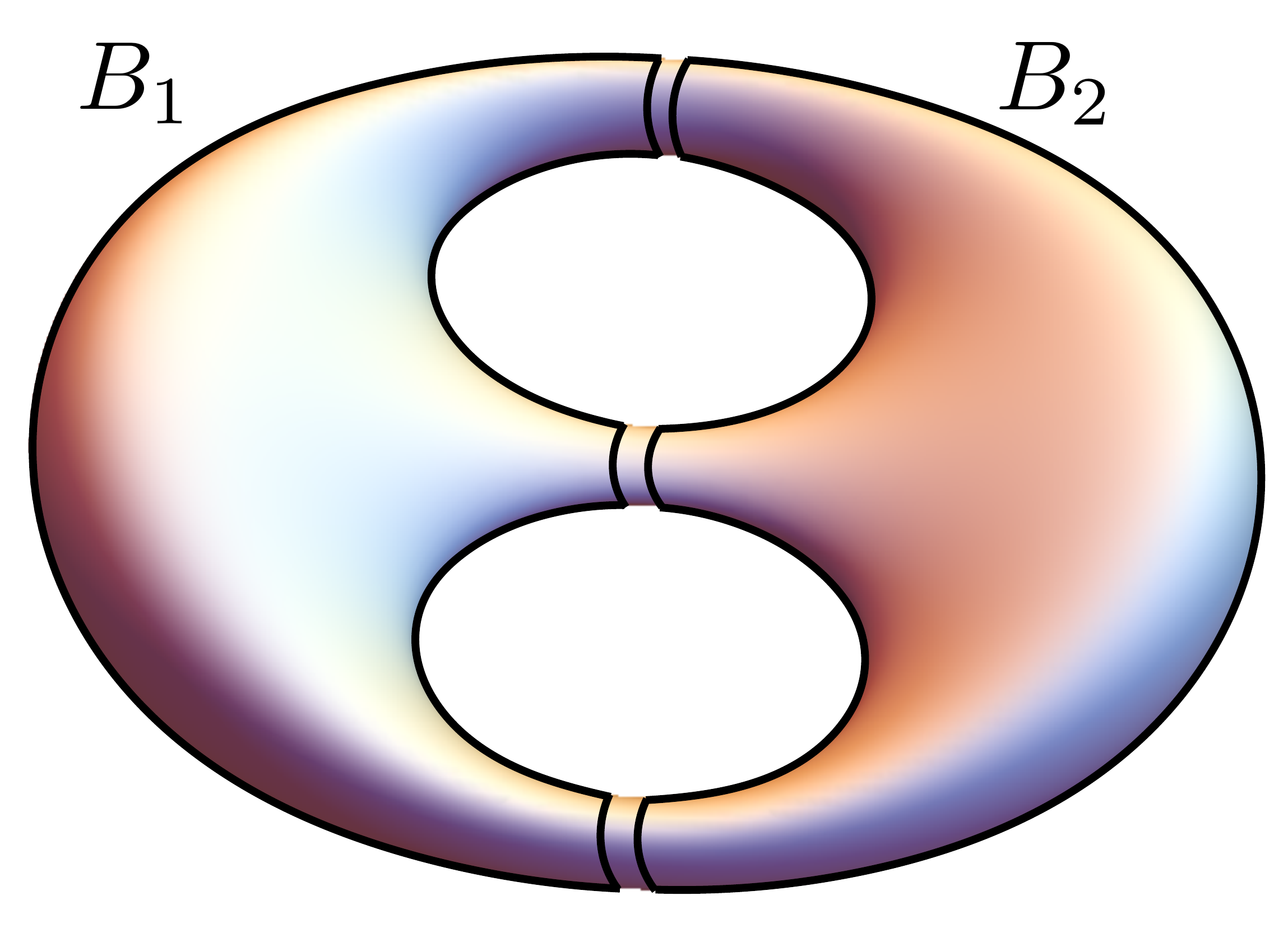}}
	\hspace{.5cm}
	\subfigure[\label{fig:h2-sigma}]{\includegraphics[width=4.5cm]{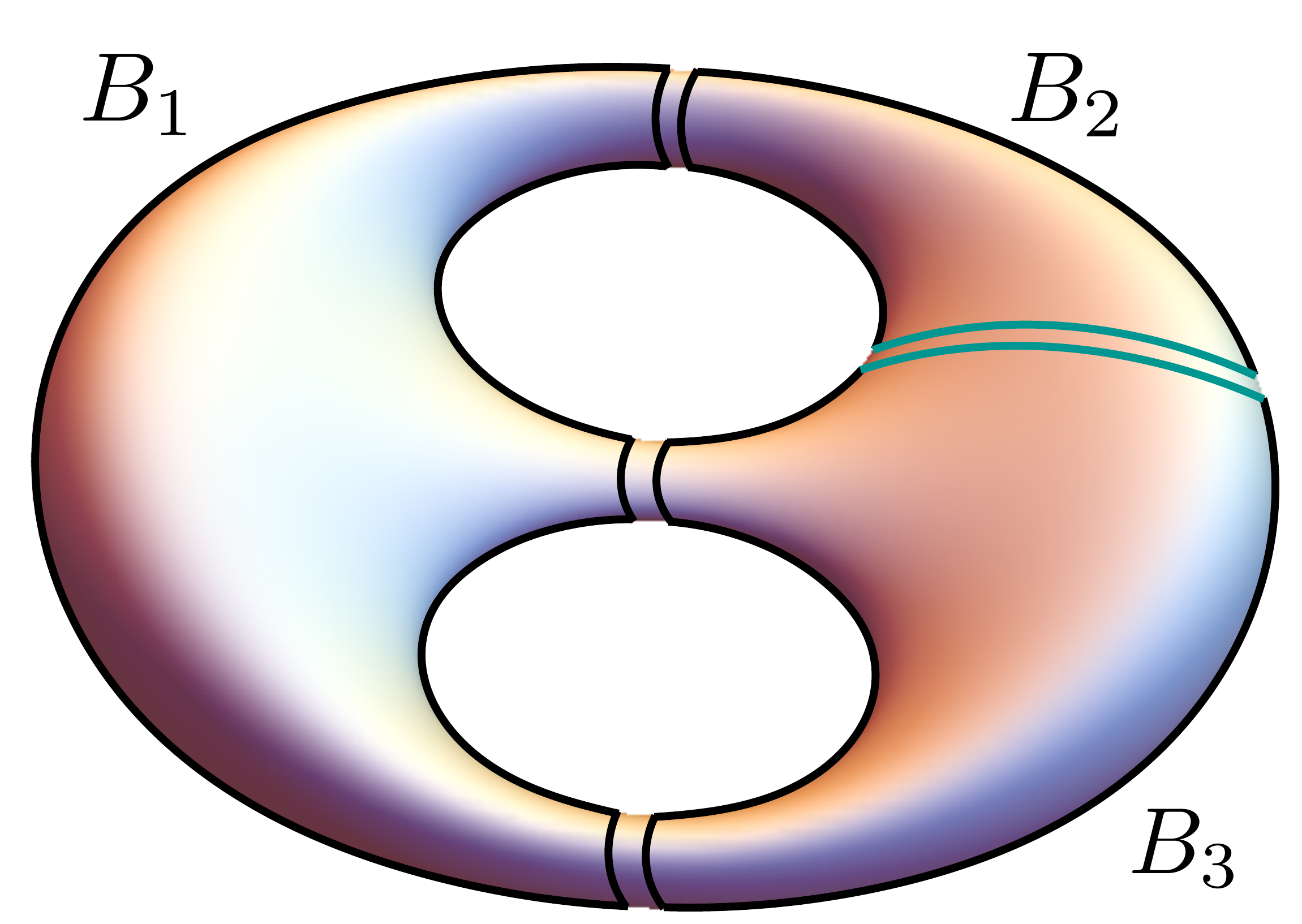}}
	\hspace{.5cm}
	\subfigure[\label{fig:h2p-sigma}]{\includegraphics[width=4.5cm]{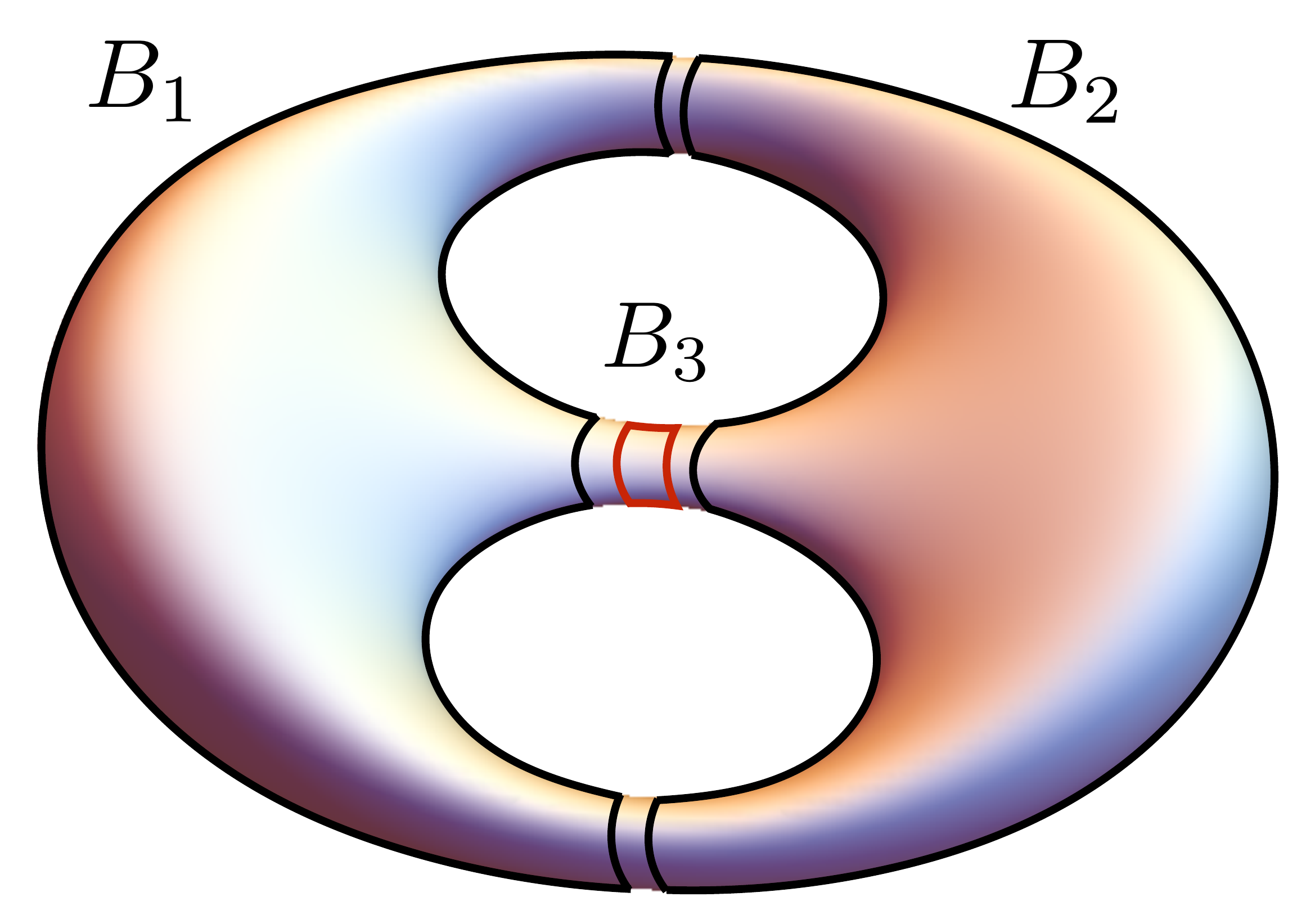}}
	
	\caption{\small A $g=2$ Riemann surface can be taken to be almost split, with three small necks connecting two pairs of pants. \subref{fig:h1-sigma}: one can fit two $B_i$ with small $\mathrm{Vol}_f( \partial B_i)/\mathrm{Vol}_f (B_i)$, so $h_1$ is small. \subref{fig:h2-sigma},\subref{fig:h2p-sigma}: one cannot fit three such $B_i$: the boundaries in green will be large or the area delimited in red will be small.}

	\label{fig:h-sigma}
\end{figure}

In the discussion above we have used that fact that, for modes that only depend on $\Sigma_g$, $\Delta_f$ reduces to the standard Laplacian on the Riemann surface, and we have applied \eqref{eq:funano-lower} to it. However, \eqref{eq:funano-lower} also applies to the full six-dimensional space with metric \eqref{eq:MNmetric} and non-trivial $\ee^f$, since it is formulated in the more general Bakry--\'Emery case. Thus, since the $\mathbb{S}^4$ does not introduce any small necks, the full spectrum does not have extra small eigenvalues either, other than the ones obtained from $\Sigma_g$. 

For compactifications on higher-dimensional hyperbolic spaces, there are generalizations of \eqref{eq:MNmetric}; for example, given a three-dimensional hyperbolic $H_3$, there is a AdS$_4\times H_3$ solution of $d=7$ gauged supergravity \cite{pernici-sezgin}, which can be uplifted to $d=11$. However, unlike for Riemann surfaces, $H_3$ cannot be taken to have arbitrarily small necks. There is a so-called ``thick--thin decomposition'', but the ``thin'' part does not disconnect $H_3$. Indeed in $n\geqslant 3$ a universal lower bound exists \cite{schoen-eigenvalue} for any manifold whose curvature is bounded between two negative constants, and in particular for any hyperbolic manifold $H_n$ with finite volume. In other words, in this case the smallest eigenvalue cannot be made arbitrarily small.

\subsubsection{Direct analysis} 
\label{ssub:direct}

We can also check this general result by directly analyzing the spectrum of $\Delta_f$. We will show that the only arbitrarily small eigenvalues come from the spectrum of $\Sigma$. We could argue for this by adapting the analysis in \cite{chen-gutperle-uhlemann}, but for completeness we prefer giving our own alternative version of the argument.

For general fluctuations on the background \eqref{eq:MNmetric}, $\Delta_f$ decomposes in
\begin{equation}\label{eq:splitOp1}
	\Delta_f = 2 \Delta_{(\Sigma_g)} - 2 \partial_{\phi}^2 + 4 y \partial_x
\partial_{\phi} + 2 \Delta_f^{(\mathbb{S}^4)}\;,
\end{equation}
where $\Delta_{(\Sigma_g)}$ is the standard Laplacian on $\Sigma_g$ and $\Delta_f^{(\mathbb{S}^4)}$ is a Bakry--\'Emery Laplacian on the (non-round) $\mathbb{S}^4$.
As anticipated above, these two operators are coupled by terms that mix the fibered $\mathbb{S}^1$ and the Riemann surface. Nevertheless, they can be decoupled by expanding the putative eigenfunction $\psi$ on a basis of functions on the fibered $\mathbb{S}^1$ as 
\begin{equation}\label{eq:qbasis}
	\psi = \psi^{(5)}_q(\Sigma,\theta,\mathbb{S}^2) \ee^{\ii q \phi}\;.
\end{equation}
Notice that, for simplicity, we have complexified $\psi$ and its eigenvalue equation. By linearity, real solutions are given by the real and imaginary parts of $\psi$. Also, in \eqref{eq:qbasis} we have used only single element of the basis on $\mathbb{S}^1$, since, invoking again linearity, all the $\psi_q^{(5)}$ will satisfy the same equation with different $q$'s.
With this definition, the operator \eqref{eq:splitOp1} becomes
\begin{subequations}\label{eq:splitOp2}
	\begin{align}
		\Delta_f^q &:= 2 \Big{(}\Delta_{(\Sigma_g)} + q^2 + 2 i q y \partial_x\Big{)} + \mathcal{D}_{(\theta, \mathbb{S}^2 ; q)}^2\\
		\label{eq:magnSchr}&= 2  \Big{(}i \nabla_{(\Sigma_g)}+ q \tilde{\xi}\Big{)}^2 + \mathcal{D}_{(\theta, \mathbb{S}^2 ; q)}^2
	\end{align}
\end{subequations}
where $\tilde{\xi}$ is the vector field dual to the 1-form connection $\xi$, $\tilde{\xi} = y^{-1}\partial_x$, and  $\mathcal{D}_{(\theta, \mathbb{S}^2 ; q)}^2$ is an operator acting on the $\mathbb{S}^2$ and the interval coordinate $\theta$, which we will make explicit soon.
The operator in the brackets in \eqref{eq:magnSchr} is a \emph{magnetic Schr\"odinger operator} on $\Sigma_g$.
It acts on sections of a U(1) bundle on $\Sigma_g$ and, since it is the square of a hermitian operator, it is also non-negative.
We call its eigenvalues $\lambda_{g}^2$.
As promised, the two operators on the right hand sides of \eqref{eq:splitOp2} are now decoupled as they act on different spaces.
Expanding $\psi^{(5)}_q$ on a basis of eigenfunctions of the magnetic magnetic Schr\"odinger operator on $\Sigma_g$ and of the standard Laplacian on the round $\mathbb{S}^2$, $\psi^{(5)}_q := \eta(\theta)Y_{\lambda_g^2}Y_{\lambda_{\mathbb{S}^2}^2}$ with eigenvalues respectively $\lambda_{g}^2$ and $\lambda_{\mathbb{S}^2}^2$, our original eigenvalue problem reads
\begin{equation}\label{eq:alphaeq}
	\left[- 2\alpha^{-2} \partial_{\alpha} (\alpha^2 (1 - \alpha^2)
	\partial_{\alpha}) -   (m^2 - q^2 - 2
	\lambda_{g}^2) + 2\left( q^2 \frac{\alpha^2}{1-\alpha^2}+\left(\frac{1+\alpha^2}{\alpha^2}\right)\lambda_{\mathbb{S}^2}^2 \right) \right] \eta = 0
\end{equation}
where we have also switched to the coordinate $\alpha := \cos(\theta)$.
We now notice that the first and last group of terms in the round brackets inside the operator in the square brackets in \eqref{eq:alphaeq} are non-negative.  This can be checked by multiplying \eqref{eq:alphaeq} by $\eta$ on the left and integrating against our measure $\sqrt{\bar{g}}\ee^f \propto \alpha^2$.  Thus, for equation \eqref{eq:alphaeq} to admit solutions, the middle term needs to be non-positive, implying the bound
\begin{equation}\label{eq:boundq}
	m^2\geqslant q^2+2\lambda_g^2 \geqslant q^2\;.
\end{equation}
Thus, states with $q\neq 0$ do not have small masses, since $q\in \mathbb{Z}$.

Let us analyze the modes with $q = 0$, i.e.~constant along $\phi$. From \eqref{eq:splitOp1} we have that $\Delta_f$ decomposes into a simple sum.
If the $\mathbb{S}^4$ does not admit small eigenvalues, then the small eigenvalues are only the ones on $\Sigma_g$.
We have already argued in Sec. \ref{ssub:bounds-Sigma} that it is possible to tune a finite number of these to be arbitrarily small, by combining theorems on decomposition of Riemann surfaces with the various isoperimetric theorems, but in the spirit of the present section we would like to compare these with direct knowledge of the spectrum.  
Luckily, direct results on the spectrum of the Laplacian on compact Riemann surfaces are available and teach us that for a Riemann surface of genus $g$ the first $2g-3$ 
eigenvalues can be arbitrarily small \cite[Th.~8.1.3]{buser-book}, with the $(2g-2)$-th always greater than $1/4$ \cite{otal-rosas}.\footnote{Recently, bounds on the spectrum of the Laplacian on Riemann surfaces (and higher dimensional closed hyperbolic manifolds) have been obtained in \cite{bonifacio}, using the bootstrap methods introduced in \cite{bonifacio-hinterbichler} for Einstein manifolds.} This agrees with the counting that follows (\ref{eq:decLap0}).
Finally, $\mathbb{S}^4$ has no small eigenvalues because it has no small necks.
To check it directly, we can focus on modes that are constant on the round $\mathbb{S}^2$, since otherwise the $\mathbb{S}^2$ eigenvalues are again of order one.
Thus we are left with equation \eqref{eq:alphaeq} for $q = \lambda_g^2=\lambda_{\mathbb{S}^2}^2 = 0$, which reads
\begin{equation}\label{eq:alphaeq0}
	- 2\alpha^{-2} \partial_{\alpha} \left(\alpha^2 (1 - \alpha^2)
	\partial_{\alpha}\right)\eta =   m^2 \eta 
\end{equation}
Defining $\hat\eta := \alpha \eta$ and $s(s+1) := \frac12 (m^2+4)$ equation \eqref{eq:alphaeq0} becomes the Legendre equation
\begin{equation}\label{eq:assLeg}
	-\partial_\alpha ((1-\alpha^2)\partial_\alpha \hat\eta) = s(s+1)\hat\eta\;,
\end{equation}
with $\alpha \in [0,1]$. 
Without imposing any condition, the general solution is given by $\hat\eta = c_1 P_s(\alpha)+c_2 Q_s(\alpha)$, where $P_s$ and $Q_s$ are, respectively,  Legendre functions of the first and second kind. However, we are interested in solutions such that $\psi \in W^{1,2}$. This means  $\int_0^1 \alpha^2 \eta^2 = \int_0^1 \hat\eta^2 < \infty$ and  $\int_0^1 \alpha^2 (\partial_\alpha\eta)^2  < \infty$, where we recall that the $\alpha^2$ factor comes from the measure $\sqrt{\bar g} \ee^f$. 
Both Legendre functions are integrable, but integrability of the derivative selects $P_s$, with $s$ odd. Thus the first non-trivial eigenvalue on the $\mathbb{S}^4$ which does not depend on the circle or the $\mathbb{S}^2$ coordinates is $m^2 = 30$, which, as anticipated from the general geometric arguments, is not small; recall that the eigenvalue on the Riemann surface can instead be made arbitrarily small.

This completes the proof that the only arbitrarily small eigenvalues come from the eigenvalues of the standard Laplacian on the Riemann surface $\Sigma_g$.

We conclude with a comment on the physical scales. The barred metric in \eqref{eq:MNmetric} is written with respect to a unit radius AdS$_5$ metric i.e.~$L^2_{\text{AdS}_5} = 1$. However, a family of eleven-dimensional solutions with arbitrary $L^2_{\text{AdS}_5}$ can be generated starting from the unit-radius solution and acting on the eleven-dimensional metric with the rescaling $\dd s^2_{11} \to L^2_{\text{AdS}_5} \dd s^2_{11}$. The effect of this action on the AdS$_5$ factor is to restore the explicit dependence on its radius, $\dd s^2_{\text{AdS}_5}\to L^2_{\text{AdS}_5}  \dd s^2_{\text{AdS}_5}$, while at the same time rescaling the barred internal metric \eqref{eq:MNmetric} as $\bar{\dd s}^2\to L^2_{\text{AdS}_5} \bar{\dd s}^2$. Expanding the Bakry--\'Emery Laplacian of the rescaled metric, we then get that the physical (i.e.~dimensionful) masses are given by 
\begin{equation}
	m^2_{\text{phys}} = \frac{m^2}{L^2_{\text{AdS}_5}}\;.
\end{equation}
This rescaling is discretized by flux quantization since (neglecting order 1 constants and factors of $\pi$)  quantization of $F_4$ requires $\int_{X_4} F_4 \sim \ell_{11}^3 N$, where $\ell_{11}$ is the eleven-dimensional Planck length and the integral is over internal 4-cycles. Since the above rescaling acts on $F_4$ as $F_4\to L_{\text{AdS}_5}^3 F_4$, the physical length scales are related to $N$ by  $L_{\text{AdS}_5} \sim N^\frac13 \ell_{11}$. 

Having reinstated the physical units allows us to check that, even if in the discussion above we have considered the geodesic lengths in the pair of pants decomposition of $\Sigma$ to be arbitrarily small, their physical lengths can be controlled by $N$ in order to suppress possible instanton effects that would arise if the lengths of small cycles get too close to $\ell_{11}$. Indeed, since $L_{\text{AdS}_5}$ multiplies the full barred metric, the physical length of a short geodesics $\gamma$ is
\begin{equation}
\ell_{\gamma,\, \text{phys}}  =  L_{\text{AdS}_5} \ell_{\gamma }\sim N^\frac13 \ell_{11} \ell_{\gamma}	\;.
\end{equation}
This can be made $\gg \ell_{11}$ in the regime $N\gg 1$, where the supergravity approximation is reliable, even if $\ell_\gamma$ is small. On the other hand, the ratio between the physical $m_k$ and $m_1$ is not affected and, as discussed above, it can be made big for $\ell_{\gamma}\ll 1$. 

The results in this subsection give another, more solid example of the mechanism in (\ref{eq:conj2}) for $g=2$; more generally, they give examples where the first few masses can be made arbitrarily small while the others remain finite. As we stressed in Section \ref{ssub:conj}, the spin-two conjectures might still be saved by the light states created by wrapped branes; in this case, they might for example be M5s wrapping the $S^4$ and the small necks of the Riemann surface.

\subsection{Untwisted non-supersymmetric examples} 
\label{sub:unt}

The discussion above, as well as the examples presented in Section \ref{sec:bl}, apply to supersymmetric compactifications. However, we can see that supersymmetry is not a necessary condition for admitting light spin-two fields by also constructing simple classes of non-supersymmetric compactifications on Riemann surfaces. The strategy is to start from a generic AdS$_{d+k}\times M_n$ solution and use the map
\begin{equation}\label{eq:AdSmap}
	\text{AdS}_{d+k}\to \text{AdS}_d\times\mathbb{H}_k/\Gamma \times M_n\;,
\end{equation}
where $\mathbb{H}_k/\Gamma$ is a hyperbolic manifold of dimension $k$ with the same Ricci curvature of the original AdS$_{d+k}$.
This procedure produces a new AdS$_d$ solution with an internal space $\mathbb{H}/\Gamma \times M_n$. More precisely, the map \eqref{eq:AdSmap} acts on a (possibly warped) AdS$_{d+k}$ solution as\footnote{See also \cite{cordova-deluca-t-dso6}, where a similar procedure has been used as a starting point to construct dS$_4$ vacua.}
\begin{equation}\label{eq:AdSmap-met}
	\dd s^2_D = \ee^{\frac14f}(\dd s^2_{\mathrm{AdS}_{d+k}}+ \bar{\dd s}^2_{n-k}) \quad\to\quad  \dd s^2_D = \ee^{\frac14f}(\dd s^2_{\mathrm{AdS}_{d}}+ \bar{\dd s}^2_{n}) \;,
\end{equation} 
where the internal barred metric is now a simple product
\begin{equation}\label{eq:decMet}
	\bar{\dd s}^2_{n} = \dd s^2_{\mathbb{H}_k/\Gamma}+ \bar{\dd s}^2_{n-k}\;. 
\end{equation} 
In this case there is no fibration since no twist is needed to preserve supersymmety and, given that $f$ is constant along the hyperbolic manifold, the Bakry--\'Emery Laplacian \eqref{eq:BELaplacian} now decomposes as
\begin{equation}\label{eq:decLap}
	\Delta_f^{(n)}=\Delta^{(\mathbb{H}_k/\Gamma)}+\Delta_f^{(n-k)}\;.
\end{equation}
Specializing to the $k=2$ case, all the above discussion on eigenvalues on Riemann surfaces applies verbatim.
However, notice that the rather crude supersymmetry breaking in \eqref{eq:AdSmap} can generate solutions that are perturbatively unstable. Albeit there are no instabilities in the spin-two spectrum, since the Bakry--\'Emery Laplacian is non-negative, there might be instabilities arising e.g.~in the scalar sector. These would be signaled by tachyons with masses below the Breitenlohner–Freedman bound \cite{breitenlohner-freedman}. Even though the scalar operators are not known in general, if also part of the scalar spectrum is preserved upon compactification on the Riemann surface the backgrounds produced by \eqref{eq:AdSmap-met} are in danger of having modes near the BF threshold, since the bound gets tighter lowering in lower (AdS) dimension: $-\frac{(d-k-1)^2}{ 4 L^2} < - \frac{(d-1)^2}{ 4 L^2}$.
Encouraging results on this side come from \cite{orlando-park}, which shows in an explicit realization of this idea that a subset of the scalars modes can be stable, but stops short of a full computation of the KK spectrum.

\subsection{Lifting the massless graviton}

As we have seen in the previous sections, in AdS compactifications that involve a negatively curved Riemann surface $\Sigma_g$, the light part of the spectrum of the Bakry--\'Emery Laplacian often coincides with the light part of the spectrum on the Riemann surface. We have focused on an explicit class of examples and argued for this effect in two ways: in Section \ref{ssub:bounds-Sigma} we used our general theorems that bound eigenvalues of $\Delta_f$ in terms of the isoperimetric constants, and in Section \ref{ssub:direct} we directly analyzed the differential eigenvalue equation, confirming that $\Delta_f$ does not have other small eigenvalues.

We will now consider the case where $\Sigma_g$ has infinite area. This will lift the constant mode (the $d$-dimensional massless graviton), providing a way to generate large classes of compactifications of String/M-theory with a finite number of light spin-two fields, just as in the finite-measure examples we saw earlier in this Section.

The spectral theory of non-compact, infinite area Riemann surfaces with negative curvature is very rich, and has seen much recent progress. In the following, we will handpick the results we need to count the number of light spin-two fields, referring to \cite[Chapter 2]{borthwick2007spectral} for a more detailed introduction. 

As in the more familiar compact case, infinite-area negatively-curved Riemann surfaces can be constructed from discrete\footnote{A discrete subgroup of PSL(2,$\mathbb{R}$) is called \emph{Fuchsian}. The discreteness requirement is needed for the quotient to be a well-defined metric space.} subgroups $\Gamma$ of PSL(2,$\mathbb{R}$), the isometry group of the hyperbolic plane $\mathbb{H}$, as $\mathbb{H}/\Gamma$. We will restrict our attention to Riemann surfaces with finite Euler characteristic (also called \emph{topologically finite}); very little is known about the spectrum of the Laplacian outside this class. Algebraically, this requirement implies that $\Gamma$ is finitely generated, and geometrically it results in the fundamental domain of its action on $\mathbb{H}$ being a finite-sided convex polygon. Moreover, it entails that there are only two types of possible ``geometries at infinity'' (or \emph{hyperbolic ends}): \emph{cusps} and \emph{funnels}. For $r\geqslant 0$, cusps have metric $\dd r^2 + \ee^{-2r} \dd \theta^2$ and make $\Sigma_g$ non-compact but leave the area finite; funnels have metric $\dd r^2 + \cosh^2 r \dd \theta^2$ and are responsible for an infinite area. We will restrict to hyperbolic surfaces with one or more funnels, but no cusps. A topologically finite hyperbolic surface with no cusps can be decomposed in a \emph{compact core} $K$ plus a finite number $n_\mathrm{f}$ of funnels \cite[Th.~2.23]{borthwick2007spectral}; see Fig.~\ref{fig:funnel}.

\begin{figure}[ht]
	\centering	
		\subfigure[\label{fig:f1-sigma}]{\includegraphics[width=5cm]{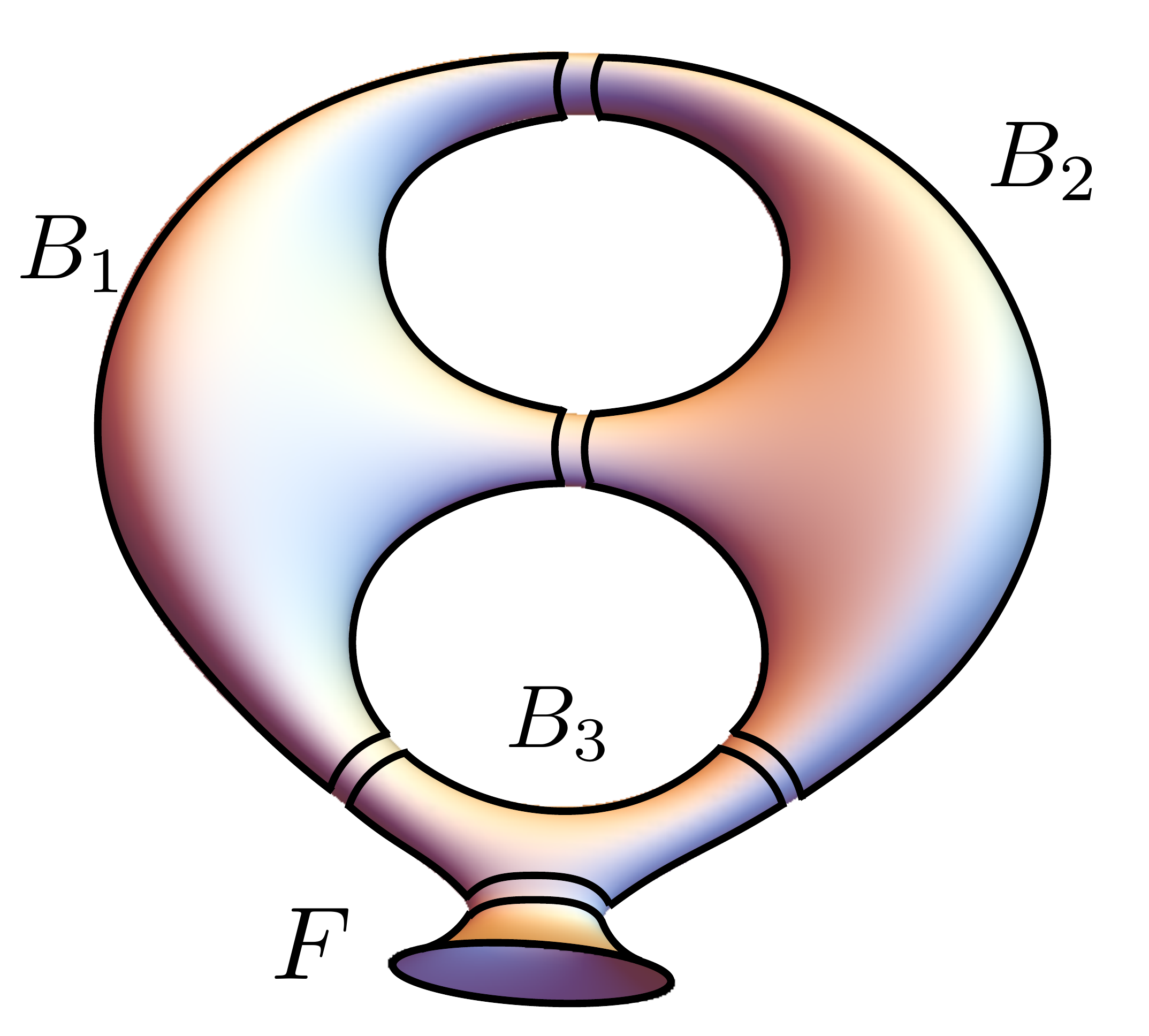}}
		\hspace{1cm}
		\subfigure[\label{fig:f2-sigma}]{\includegraphics[width=4cm]{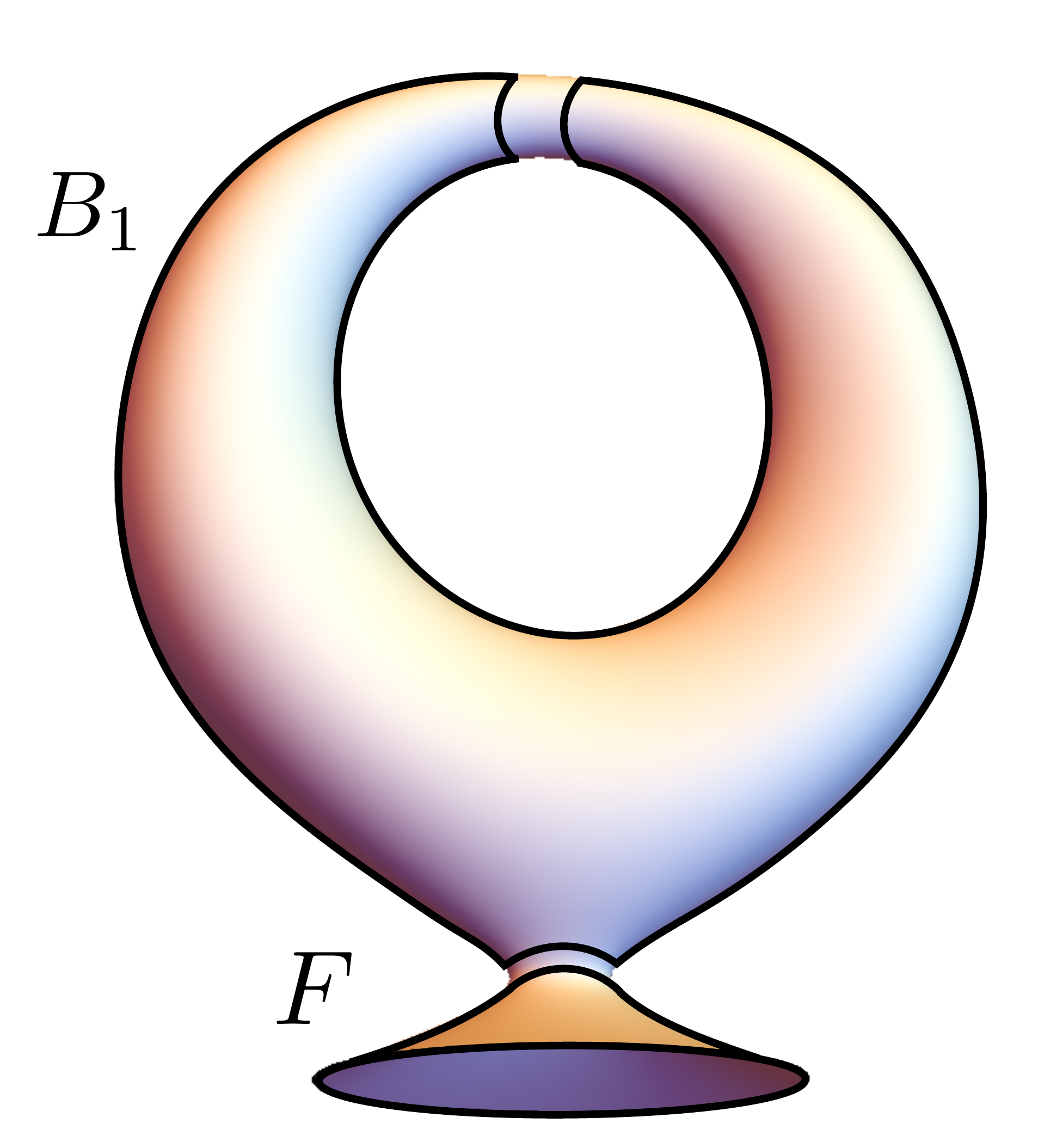}}
		\caption{\small Hyperbolic surfaces with a single funnel $F$. There is a decomposition in a compact core $K$ (the union of the pairs of pants $B_i$) and the funnel, joined along a geodesic. \subref{fig:f1-sigma} a $g=2$ example with 3 $B_i$ with small necks. \subref{fig:f2-sigma} a $g=1$ surface constructed gluing two legs of a single pair of pants together.}
		\label{fig:funnel}
	\end{figure}

The first important remark is that this class of examples introduces a continuous part of the spectrum. For any surface with hyperbolic ends, the bottom of the essential spectrum of the Laplacian is $1/4$, with the discrete part consisting of finitely many eigenvalues in the range $\left(0,1/4\right)$ \cite[Theorem 7.1]{borthwick2007spectral}.\footnote{If the area is finite, so with only cusps but no funnels, above $1/4$ there might also be embedded eigenvalues, namely eigenvalues in the continuous spectrum. Here we are not interested in this possibility.} This is a bit similar to the $p=5$ case discussed in general in Section \ref{ssub:continuous} and appearing in the examples analyzed in Section \ref{sub:asas}. It would be interesting to know whether a (gapped) continuous spectrum is also present in other massive gravity models obtained from String/M-theory.

Focusing now on the discrete spectrum, fortunately the results we quoted after \eqref{eq:boundq} admit non-compact generalizations. 
It was indeed proven in \cite[Th.~1.5]{ballmann-matthiesen-mondal-15} (extending \cite{otal-rosas}) that on a Riemann surfaces of finite type, i.e.~with finite Euler characteristic $\chi=2-2g-n_\mathrm{f}$ and compact boundary, the $(-\chi+1)$-th\footnote{When comparing with the results in Sec.~\ref{sub:mn} notice that in the compact case we do not usually count the zero-mode, and thus in that case ``$k$-th'' means ``$k$-th non-trivial''. Since in this section we are allowing the bottom of the spectrum to be non-zero, we are counting it separately.} eigenvalue is larger than $1/4$. In particular, in the case we are interested in, where the only hyperbolic ends are funnels, there are only $-\chi$ discrete eigenvalues. 

Moreover, the proof of \cite[Th.~8.1.3]{buser-book} is in fact still valid. The min-max principle is discussed there for the compact case, but in fact it also applies to the discrete eigenvalues below the continuum, by \cite[Th.~4.5.2]{Dav}.

Let us now see how these results agree with those of Section \ref{sec:bounds}. Consider first the smallest eigenvalue.
By construction, the compact core $K$ has $n_\mathrm{f}$ geodesic boundaries and the same genus of $\Sigma_g$, and thinking of it as a Riemann surface with boundaries, we can decompose it in pairs of pants, as we did in Section \ref{ssub:bounds-Sigma} for the full $\Sigma_g$. As we saw there, the lengths of the boundary geodesics are independent; so we can arrange for at least one $B$ to have very small $\frac{\text{Per}(B)}{\mathfrak m(B)}$, by taking it to be a pair of pants with very small boundary geodesics. By (\ref{eq:buser0}), we can now conclude that the smallest eigenvalue $\lambda_1$ can be made arbitrarily small.

For the higher eigenvalues, we consider Th.~\ref{th: RCDUpper}.
The definition of the Cheeger constant \eqref{def: m-Cheeger} instructs us to only consider Borel subsets with finite measure, and thus we again focus on the core $K$ and compute how many $B_i$ with a small $\frac{\text{Per}(B_i)}{\mathfrak m(B_i)}$ we can fit into it. 
Since each pair of pants has $\chi=-1$, their total number is $-\chi=2-2g-n_\mathrm{f}$.
Thus we can cut $2g-2+n_\mathrm{f}$ pieces with arbitrarily small $\frac{\text{Per}(B_i)}{\mathfrak m(B_i)}$, $ i = 1,\dots,2g-2+n_\mathrm{f}$; we have again used that the lengths of the boundary geodesic of a pair of pants can be arbitrary real numbers.
All in all, this means that we can make the $(-\chi-1)$-th Cheeger constant $h_{-\chi-1}$ small with the same strategy as in Section \ref{ssub:bounds-Sigma}. Now Th.~\ref{th: RCDUpper} implies that the $-\chi$ eigenvalues $\lambda_0,\dots,\lambda_{-\chi-1}$ are allowed to be arbitrarily small.

For a simple example, consider the case $g = 1$, $n_\mathrm{f} = 1$, depicted in Fig.~\ref{fig:f2-sigma}. This surface is hyperbolic (see e.g.~\cite[Th.~27.12]{foster}) and can be thought of as a single pair of pants with two legs glued to each other, and the remaining leg glued to a funnel. In this case, there is a single neck that can be made arbitrarily small, resulting in a single arbitrarily light spin-two field and no massless graviton.

In conclusion, the models in this subsection provide a method to construct AdS massive gravity models from String/M-theory, with the aforementioned caveat about the presence of a (gapped) continuum. 


\section*{Acknowledgements}

We thank C.~Bachas and E.~Palti for discussions. GBDL is supported in part by the Simons Foundation Origins of the Universe Initiative (modern inflationary cosmology collaboration) and by a Simons Investigator award. AM is supported by the ERC Starting Grant 802689 ``CURVATURE''. AT is supported in part by INFN and by MIUR-PRIN contract 2017CC72MK003.

\bibliography{at}
\bibliographystyle{at}

\end{document}